\pdfoutput=1

\documentclass[envcountsame]{llncs}
\usepackage[utf8]{inputenc}
\usepackage{amsmath}

\usepackage{multirow}
\usepackage{amssymb}
\usepackage{stmaryrd}
\usepackage{booktabs}
\usepackage{xspace}
\usepackage{placeins}
\usepackage{enumerate}
\usepackage{verbatim}
\usepackage{enumerate}
\usepackage{mathtools}
\usepackage{url}
\usepackage{thm-restate}
\usepackage{color}
\usepackage{subfigure}
\usepackage{wrapfig}

\usepackage[bookmarks=false]{hyperref}

\newcommand\inter[1]{\llbracket #1 \rrbracket}
\newcommand\trans[2]{\ensuremath{#1\mid #2}}
\newcommand\delay{\Delta}
\newcommand\tra{\mathcal{T}}
\newcommand\eof{\hfill$\square$}

\usepackage{graphicx}

\usepackage{xcolor,pgf,tikz}
\usetikzlibrary{arrows,automata,decorations.pathmorphing,snakes}

\title{
Multi-Sequential Word Relations}

\author{
Ismaël Jecker\and
Emmanuel Filiot
}
\institute{
Université Libre de Bruxelles}

\begin{document}
\maketitle

\begin{abstract}
    Rational relations are binary relations of finite words that are
    realised by non-deterministic finite state transducers (NFT). 
    A particular kind of rational relations is the sequential
    functions. Sequential functions are the functions that can be
    realised by input-deterministic transducers. Some rational
    functions are not sequential. However, based on a property on
    transducers called the twinning property, it is decidable in
    \textsc{PTime} whether a rational function given by an NFT is
    sequential. In this paper, we investigate the
    generalisation of this result to multi-sequential relations,
    i.e. relations that are equal to a finite union of sequential
    functions. We show that given an NFT, it is decidable in
    \textsc{PTime} whether the relation it defines is
    multi-sequential, based on a property called the weak twinning
    property. If the weak twinning property is satisfied, we give a procedure that effectively
    constructs a finite set of input-deterministic transducers whose
    union defines the relation. This procedure generalises to
    arbitrary NFT the determinisation procedure of 
    functional NFT.
\end{abstract}


Finite transducers extend finite automata with 
output words on transitions. Any successful computation (called run) of a
transducer defines an output word obtained by concatenating, from
left to right, the words occurring along the transitions of that
computation. Since transitions are non-deterministic in general, there
might be several successful runs on the same input word $u$, and hence 
several output words associated with $u$. Therefore, finite
transducers can define binary relations of finite words, the so-called
class of rational relations \cite{Elgot:Mezei:ibmjrd:1965,berstel2009}.
Unlike finite automata, the equivalence problem, i.e. whether two transducers define the same
relation, is undecidable \cite{DBLP:journals/jacm/Griffiths68}. This has
motivated the study of different subclasses of rational relations, and
their associated definability problems: given a finite transducer $T$,
does the relation $\inter{T}$ it defines belong to a given class
$\mathcal{C}$ of relations? We survey the most
important known subclasses of rational relations.

\paragraph{Rational Functions} An important subclass of rational relations is the class of
\emph{rational functions}.  It enjoys decidable
equivalence and moreover, it is decidable whether a transducer
is functional, i.e. defines a function. This latter result was first shown by 
Sch\"utzenberger with polynomial space complexity \cite{Schutz75} and the complexity has been refined to
polynomial time in \cite{GurIba83,BealCPS03}.

A subclass of rational functions which enjoys good algorithmic
properties with respect to evaluation is the class of \emph{sequential
  functions}. Sequential functions are those functions defined by 
finite transducers whose underlying input automaton is deterministic
(called sequential transducers). Some rational functions are not
sequential. E.g., over the alphabet $\Sigma = \{a,b\}$, the 
function $f_{swap}$ mapping any word of the form $u\sigma$ to $\sigma u$, where
$u\in\Sigma^*$ and $\sigma\in\Sigma$, is rational but not sequential,
because finite transducers process input words from left-to-right, and
therefore any transducer implementing that function should
 guess non-deterministically the last symbol of $u\sigma$. 
Given a functional transducer, it is decidable whether it defines a
sequential function \cite{DBLP:journals/tcs/Choffrut77}, based on a
structural property of finite transducers called the \emph{twinning
  property}. This property can be decided in \textsc{PTime}
\cite{BealCPS03} and therefore, deciding whether a functional
transducer defines a sequential function is in
\textsc{PTime}. If the twinning property holds, one can
``determinise'' the original transducer into an equivalent sequential
transducer.

It turns out that many examples of rational functions from the
literature which are not sequential are \emph{almost} sequential, in
the sense that they are equal to a finite union of sequential
functions. Such functions are called \emph{multi-sequential}. 
For instance, the function $f_{swap}$ is multi-sequential, as
$f_{swap} = f_a\cup f_b$, where $f_a$ is
the partial sequential function mapping all words of the form $ua$ to
$au$ (and similarly for $f_b$). Some rational functions are not
multi-sequential,  such as functions that are iterations of non-sequential functions. E.g., the function 
mapping $u_1\# u_2\#\dots \# u_n$ to $f_{swap}(u_1)\dots
f_{swap}(u_n)$ for some separator $\#$, is not multi-sequential. 
Multi-sequential functions have been
considered by Choffrut and
Sch\"utzenberger in \cite{DBLP:conf/stacs/ChoffrutS86}, where it is shown
that multi-sequentiality for functional transducers is a decidable
problem. 

\paragraph{Contribution} In this paper, we investigate multi-sequential relations,
i.e. relations that are equal to a finite union of sequential functions. Our
main result shows that, given a finite transducer, it is decidable in
\textsc{PTime} whether the relation it defines is
multi-sequential. Our procedure is based on a simple characterisation
of multi-sequential relations by means of a structural property,
called the weak twinning property (WTP), on finite transducers. We
show that a finite transducer defines a multi-sequential relation iff
it satisfies the WTP. We define a ``determinisation'' procedure of
finite transducers satisfying the WTP, that decomposes them into finite
unions of sequential transducers. Finally, we also investigate the
computational properties of multi-sequential relations and show, that,
for a natural computational model for word relations, 
multi-sequential relations correspond to the relations that
can be evaluated with constant memory.

\paragraph{Related Works} As already mentioned, multi-sequential
functions were considered in \cite{DBLP:conf/stacs/ChoffrutS86}. Our
weak twinning property is close to the characterisation of
multi-sequential functions of \cite{DBLP:conf/stacs/ChoffrutS86}, which is based on analysing 
families of \emph{branching paths} in transducers. Thanks to the
notion of \emph{delay} between output words, our property is simpler
and can be decided in \textsc{PTime}, for arbitrary (not
necessarily functional) transducers. Compared to
\cite{DBLP:conf/stacs/ChoffrutS86}, our decomposition procedure is
more constructive. It extends the known determinisation procedure of
functional transducers, to multi-sequential relations, and applies
directly on arbitrary finite state transducers (in
\cite{DBLP:conf/stacs/ChoffrutS86}, the functional transducers are assumed to be
unambiguous, but removing ambiguity is worst-case exponential).

Finite-valued rational relations are rational relations such that any
input word has at most $k$ images by the relation, for a fixed
constant $k$. Finite-valuedness (existence of such a $k$) is decidable in \textsc{PTime} for
rational relations, and any $k$-valued rational relation is equal to a
union of $k$ rational functions \cite{journals/mst/SakarovitchS10,SICOMP::Weber1993,DBLP:journals/acta/Weber89}. 
Equivalence of finite-valued rational relations is decidable
\cite{GurIba83}. Clearly, any multi-sequential relation is
finite-valued. However, for the multi-sequentiality problem, it would
not have been simpler to assume that the input relation is given as a
finite union of rational functions, because obtaining such as
decomposition is costly, and moreover, our property is simpler than the structural properties 
of \cite{journals/mst/SakarovitchS10,DBLP:journals/acta/Weber89} that 
characterise finite-valuedness.

Finally, finitely-sequential relations have been considered in
\cite{DBLP:journals/ijfcs/AllauzenM03}. They correspond to relations
that can be realised by an input-deterministic transducer whose
accepting states can, at the end of the run, output additional words
from a finite set. Such relations are much weaker than
multi-sequential relations.

The known and new results of this paper are summarised in
Table~\ref{table:summary}.

\begin{table}[t]
    \centering
    \begin{tabular}{c|c|c|c|c}
\toprule
sequentiality & multi-sequentiality  & functionality & multi-sequentiality & finite-valuedness \\
 & (for functions) &  & 
\\\midrule
\textsc{PSpace} \cite{DBLP:journals/tcs/Choffrut77} & \textsc{Decidable} \cite{DBLP:conf/stacs/ChoffrutS86} & \textsc{PSpace}
                                                  \cite{Schutz75} &
                                                           \textsc{PTime}
                                                           (\textbf{this
      paper}) & \textsc{PTime} \cite{journals/mst/SakarovitchS10,DBLP:journals/acta/Weber89} \\
\textsc{PTime} \cite{BealCPS03} & \textsc{PTime} (\textbf{this paper}) & \textsc{PTime}
                                                  \cite{GurIba83,BealCPS03}&
                                                           & \\
\bottomrule
\end{tabular}
\caption{\label{table:summary}Definability problems for rational relations given by finite transducers}\label{tab}
\vspace{-9mm}
\end{table}


\vspace{-4mm}
\section{Rational Word Relations}
\vspace{-2mm}

We denote by $\mathbb{N}$ the set of natural numbers
$\{0,1,\dots\}$, and by $\mathcal{P}(A)$ the set of subsets of
a set $A$, and by $\mathcal{P}_f(A)$ the set of finite subsets of $A$.

\vspace{2mm}
\noindent \textbf{Words and delays} Let $\Sigma$ be a (finite)
alphabet of symbols. The elements of the free monoid $\Sigma^*$ are
called words over $\Sigma$. The length of a word $w$ is written $|w|$.
The free monoid $\Sigma^*$ is partially ordered by the prefix relation
$\leq$. We denote by $\Sigma^{-1}$ the set of
symbols $\sigma^{-1}$ for all $\sigma\in\Sigma$. Any word $u\in
(\Sigma\cup \Sigma^{-1})^*$ can be reduced into an irreducible word $\overline{u}$ by the equations
$\sigma\sigma^{-1} = \sigma^{-1}\sigma = \epsilon$ for all
$\sigma\in\Sigma$. Let $G_\Sigma$ be the set of irreducible words over
$\Sigma\cup \Sigma^{-1}$. The set $G_\Sigma$ equipped with concatenation $u.v =
\overline{uv}$ is a group, called the free group over $\Sigma$. We denote by $u^{-1}$ the inverse
of $u$. E.g. $(a^{-1}bc)^{-1} = c^{-1}b^{-1}a$. The \emph{delay} between two words $v$ and $w$ is the element
$\Delta(v,w) := v^{-1} w \in G_\Sigma$. E.g., $\Delta(ab,acd) =
b^{-1}cd$.

\vspace{2mm}
\noindent \textbf{Finite automata} A (finite) automaton
over a monoid $M$ is a tuple $\mathcal{A} = (Q,E,I,T)$ where $Q$ is the finite set
of states, $I \subset Q$ is the set of initial states, $T \subset Q$
is the set of final states, and $E \subset Q \times M \times Q$ is the
finite set of edges, or transitions, labelled by elements of the monoid $M$.

For all transitions $e = (q_1,m,q_2) \in E$, $q_1$ is called the source of $e$, $q_2$ its target and $m$ its label.
A run of an automaton is a sequence of transitions $r = e_1 \ldots
e_n$ such that for every $1 \leq i < n$, the target of $e_i$ is equal
to the source of $e_{i+1}$. We write $p \xrightarrow{m}_{\mathcal{A}}
q$ (or just $p \xrightarrow{m} q$ when it is clear from the context) to mean
that there is a run $e_1\dots e_n$ such that $p$ is the source of
$e_1$, $q$ the target of $e_n$, and $m$ is the product of the labels
of the $e_i$. A run is called accepting if its source is an initial state and its target is a final state.
An automaton is called trim if each of its states occurs in at least
one accepting run. It is well-known that any automaton can be trimmed
in polynomial time. The language recognised by an automaton over a monoid $M$ is the set of elements of $M$ labelling its accepting runs.
A $\Sigma$-automaton is an automaton over the free monoid $\Sigma^*$ such that each edge is labelled by a single element $\sigma$ of $\Sigma$.
A $\Sigma$-automaton is called deterministic if it has a single initial state,
and for all $q \in Q$ and $\sigma \in \Sigma$, there exists at most one
transition labelled $\sigma$ of source $q$.

Given two automata $A_1$ and $A_2$ over a monoid $M$, their disjoint
union $A_1\cup A_2$ is defined as the disjoint union of their set of states, initial and final states,
and transitions. It recognises the union of the their respective
languages. 

\vspace{2mm}
\noindent \textbf{Finite transducers} Let $\Sigma$ and $\Gamma$ two
alphabets. A (finite)
transducer $\mathcal{T}$ from $\Sigma^*$ to $\Gamma^*$ is a tuple
$(Q,E,I,T,f_T)$ such that $(Q,E,I,T)$ is a
finite automaton over the monoid $\Sigma^*\times \Gamma^*$, called the underlying automaton of $\mathcal{T}$,
such that $E\subseteq Q\times \Sigma\times \Gamma^*\times Q$, 
and $f_T :  T \rightarrow \Gamma^*$ is the final output
function. In this paper, the input and output alphabets are always denoted by
$\Sigma$ and $\Gamma$. Hence, we just
use the terminology transducer instead of transducer from $\Sigma^*$
to $\Gamma^*$.

A run (resp. accepting run) of $\tra$ is a run (resp. accepting run)
of its underlying automaton. We write $p\xrightarrow{u|v} q$ instead
of $p\xrightarrow{(u,v)} q$. The relation recognised by $\tra$ is the
set $\inter{\tra}$ of pairs $(v,wf_T(t))$ such that $i
\xrightarrow{v|w} t$ for $i\in I$ and $t\in T$. A transducer $\tra$ is functional if $\inter{\tra}$ is a function. It is called trim if its underlying automaton is trim.
The input automaton of a transducer is the $\Sigma$-automaton obtained by projecting the labels of its underlying automaton on their first component.
Given a transducer $\mathcal{T}$, we denote by $M_{\mathcal{T}}$ the
maximal integer $|v|$, $v\in \Gamma^*$, such that $v$ labels a
transition of $\tra$ or $v = f_T(q)$ for some $q\in Q$. 

A rational transducer is defined as a transducer, except that its
transitions are labeled in $\Sigma^*\times \Gamma^*$. Rational
transducers are strictly more expressive than transducers and define
the class of rational relations. E.g., by using loops 
$q\xrightarrow{\epsilon|w} q$, a word can have infinitely many images
by a rational relation. If there is no such loop, it is
easily shown that rational transducers are equivalent to transducers\footnote{Transducers are sometimes called real-time in the
  literature, and rational transducers just transducers. To avoid
  unecessary technical difficulties, we establish our main result for 
  (real-time) transducers, but, as shown in Remark \ref{rm:realtime},
  it still holds for rational transducers.}.


\vspace{-3mm}
\section{Multi-Sequential Relations}
\vspace{-2mm}

In this section, we define multi-sequential relations, and give a
decidable property on transducers that characterise them.

\vspace{-4mm}
\subsection{From sequential functions to multi-sequential relations}
\vspace{-2mm}

A transducer $\mathcal{T} = (Q, E, I, T, f_T)$ is
\emph{sequential} if its input automaton is deterministic. A function $f:A^*\rightarrow B^*$ is
\emph{sequential} if it can be realised by a sequential
transducer\footnote{These functions are sometimes called subsequential
in the literature. We follow the terminology of
\cite{DBLP:books/daglib/0023547}.}, i.e. $f = \inter{\mathcal{T}}$ for
some sequential transducer $\mathcal{T}$. 

\begin{figure}[t]
\centering
\subfigure[$\tra_{a}$\label{ex1}]{
\begin{tikzpicture}[->,>=stealth',auto,node distance=2cm,thick,scale=0.7,every node/.style={scale=0.7}]
  \tikzstyle{every state}=[fill=yellow!30,text=black]
  \tikzstyle{initial}=[initial by arrow, initial where=left, initial text=]
  \tikzstyle{accepting}=[accepting by arrow, accepting where=right, accepting text=$b$]

  \node[initial,state] (A)  {$q_0$};
  \node[accepting,state] (B) [below of=A] {$q_1$};
  \path (A) edge node {\trans{a}{\epsilon}} (B);
  \path (B) edge [loop below] node {\trans{a}{a}} (B);
\end{tikzpicture}
}%
\subfigure[$\tra_{blank}$\label{ex2}]{
\begin{tikzpicture}[->,>=stealth',auto,node distance=2cm,thick,scale=0.7,every node/.style={scale=0.7}]
  \tikzstyle{every state}=[fill=yellow!30,text=black]
  \tikzstyle{initial}=[initial by arrow, initial where=left, initial text=]
  \tikzstyle{accepting}=[accepting by arrow, accepting where=right, accepting text=]

  \node[initial,accepting,state] (A)  {$q_0$};
  \node[accepting,state] (B) [below of=A] {$q_1$};
  \path (A) edge [bend left] node {\trans{b}{\epsilon}} (B);
  \path (B) edge [bend left] node {\trans{a}{ba}} (A);
  \path (A) edge [loop above] node {\trans{a}{a}} (A);
  \path (B) edge [loop below] node {\trans{b}{\epsilon}} (B);
\end{tikzpicture}
}%
\subfigure[$\tra_{swap}$\label{ex3}]{
\begin{tikzpicture}[->,>=stealth',auto,node distance=2cm,thick,scale=0.7,every node/.style={scale=0.7}]
  \tikzstyle{every state}=[fill=yellow!30,text=black]
  \tikzstyle{initial}=[initial by arrow, initial where=left, initial text=]
  \tikzstyle{accepting}=[accepting by arrow, accepting where=right, accepting text=]

  \node[initial,state] (A) at (0,0) {$q_0$};
  \node[state] (B) at (1,1) {$q_1$};
  \node[state] (C) at (1,-1) {$q_2$};
  \node[state,accepting] (D) at (2,0) {$q_3$};

  \path (A) edge [bend left] node {\trans{a}{aa}} (B);
  \path (B) edge [bend left] node {\trans{a}{\epsilon}} (D);
  \path (A) edge [below left,bend right] node {\trans{a}{ba}} (C);
  \path (C) edge [below right,bend right] node {\trans{b}{\epsilon}} (D);
  \path (B) edge [loop above] node {\trans{a}{a}} (B);
  \path (C) edge [loop below] node {\trans{a}{a}} (C);
\end{tikzpicture}
}%
\subfigure[$\tra^*_{swap}$\label{ex4}]{
\begin{tikzpicture}[->,>=stealth',auto,node distance=2cm,thick,scale=0.7,every node/.style={scale=0.7}]
  \tikzstyle{every state}=[fill=yellow!30,text=black]
  \tikzstyle{initial}=[initial by arrow, initial where=left, initial text=]
  \tikzstyle{accepting}=[accepting by arrow, accepting where=right, accepting text=]

  \node[initial,state] (A) at (0,0) {$q_0$};
  \node[state] (B) at (1,1) {$q_1$};
  \node[state] (C) at (1,-1) {$q_2$};
  \node[state,accepting] (D) at (2,0) {$q_3$};

  \path (A) edge [bend left] node {\trans{a}{aa}} (B);
  \path (B) edge [bend left] node {\trans{a}{\epsilon}} (D);
  \path (A) edge [below left,bend right] node {\trans{a}{ba}} (C);
  \path (C) edge [below right,bend right] node {\trans{b}{\epsilon}} (D);
  \path (B) edge [loop above] node {\trans{a}{a}} (B);
  \path (C) edge [loop below] node {\trans{a}{a}} (C);
  \path (D) edge [above] node {\trans{\#}{\#}} (A);
\end{tikzpicture}
}
\vspace{-4mm}
\caption{Finite transducers}
\label{fig:transducers} 
\vspace{-6mm}
\end{figure}
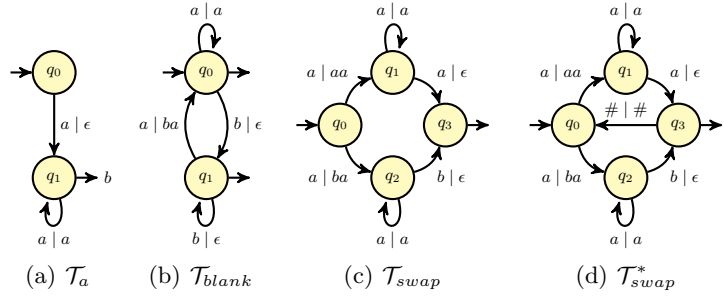

Let $\Sigma = \Gamma = \{a,b\}$. Fig. \ref{fig:transducers} depicts transducers 
that implement sequential and non-sequential functions. All states
which is the target of a source-less arrow are initial, and those
which are the source of an arrow without target, or whose target is a
word, are accepting.  The function $f_{a}$ that maps any word of the
form $a^n$, $n>0$, to $a^{n-1}b$, is sequential. It is realised by the
transducer of Fig.~\ref{ex1}. The function $f_{blank}$ replaces each
block of consecutive $b$ by a single $b$, but the last
one. E.g. $f_{blank}(abbbab) = aba$. It is sequential and defined by
the sequential transducer of Fig.~\ref{ex2}. The function $f_{swap}$
maps any word of the form $a^n\sigma$ to $\sigma a^n$, for
$\sigma\in\Sigma$. It is not sequential, because the transducer
has to guess the last symbol $\sigma$. It can be defined by the
transducer of Fig.~\ref{ex3}.

Sequential functions have been characterised by a structural property
of the transducers defining them, called the \emph{twinning
  property}.  Precisely, a trim transducers with initial state $q_0$
is twinned iff for all states $q_1,q_2$, all words $u,v\in\Sigma^*$ and
$u_1,v_1,u_2,v_2\in\Gamma^*$, if $q_0\xrightarrow{u |
    u_{1}}
q_1\xrightarrow{v | v_1} q_1$ and $q_0\xrightarrow{u | u_2}
q_2\xrightarrow{v | v_2} q_2$, then $\delay(u_1,u_2) =
\delay(u_1u_2,v_1v_2)$. E.g., by taking $u = v = a$, $u_1 =
aa$, $u_2 = ba$ and $v_1 = v_2 = a$, it is easy to see that 
the transducer of Fig.\ref{ex3} is not twinned.

\begin{theorem}[\cite{DBLP:journals/tcs/Choffrut77,BealCPS03}]
Let $\tra$ be a trim transducer.
\vspace{-2mm}
\begin{enumerate}
    \item $\tra$ is twinned iff $\inter{\tra}$ is sequential. 
    \item It is decidable in \textsc{PTime} whether a trim transducer
      is twinned. 
\end{enumerate}
\end{theorem}

The following result is a folklore result that we show in Appendix for the
sake of completeness. It states that the difference
(the delay) between the outputs of 
two input words is linearly bounded by the difference of their
input words. 
  
\vspace{-1mm}
\begin{proposition}\label{bounded_variation}
Let $\mathcal{D}$ be a sequential transducer. For all pairs $(u_1,v_1),(u_2,v_2)\in\inter{\mathcal{D}}$, $|\Delta(v_1,v_2)| \leq M_{\mathcal{D}} (|\Delta(u_1,u_2)|+2)$.
\end{proposition}

\vspace{-4mm}
\paragraph{Multi-sequential relations} The function $f_{swap}$ is not
sequential, but it is multi-sequential, in the sense that it is the
union of two sequential functions $f_1,f_2$ such that $f_1$ is the
restriction of $f_{swap}$ to words in $a^+a$ and $f_2$ its restriction
to words in $a^+b$. Precisely:

\begin{definition}[Multi-sequential relations]
    A relation $R \subseteq \Sigma^*\times \Gamma^*$ is
    multi-sequential if there exist $k$ sequential functions
    $f_1,\dots,f_k$ such that $R = \bigcup_{i=1}^k f_i$. 
\end{definition}

The \emph{multi-sequentiality problem} asks, given a transducer $\tra$,
whether $\inter{\tra}$ is multi-sequential. It should be clear that
the answer to this problem is not always positive. Indeed, even some
rational functions are not multi-sequential. It is the case
for instance for the function $f^*_{swap}$ that maps any word of the
form $u_1\#u_2\#\dots \#u_n$ to $f_{swap}(u_1)\#f_{swap}(u_2)\#\dots
\# f_{swap}(u_n)$, where $u_i\in\Sigma^*$ and $\#\not\in\Sigma$ is a
fresh symbol. This function is rational, as it can be defined by the
transducer of Fig.~\ref{ex4}. In this paper, we investigate the intrinsic reasons
making a rational relation like $f_{swap}$ multi-sequential and a
rational relation like $f_{swap}^*$ not. In particular, we define a
weaker variant of the twinning property that characterises the
multi-sequential relations by structural properties of the transducers
which define them. 

\vspace{-5mm} 
\subsection{Weak Twinning Property}
\vspace{-1mm} 

\begin{definition}\label{def:wtp}
    Let $\mathcal{T}$ be trim transducer and $q_1,q_2$ be two
    states of $\mathcal{T}$. We say that $q_1$ is weakly twinned to
    $q_2$ if for all words
    $u,v\in \Sigma^*$ and all words $u_1,u_2,v_1,v_2\in \Gamma^*$, if $q_1\xrightarrow{u\mid u_1} q_1 \xrightarrow{v\mid v_1}
    q_1\xrightarrow{u\mid u_2} q_2 \xrightarrow{v\mid v_2} q_2$, or graphically
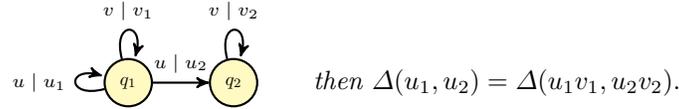
\begin{center}
\vspace{-1mm}
\begin{tikzpicture}[->,>=stealth',auto,node distance=2cm,thick,scale=0.7,every node/.style={scale=0.7}]
  \tikzstyle{every state}=[fill=yellow!30,text=black]
  \tikzstyle{initial}=[initial by arrow, initial where=left, initial text=]
  \tikzstyle{accepting}=[accepting by arrow, accepting where=right, accepting text=]

  \node[state] (A)  {$q_1$};
  \node[state] (B) [right of=A] {$q_2$};
  \tikzstyle{every node}=[font=\normalsize]
  \node[] (C) [right of=B,xshift=1.5cm] {then $\delay(u_1,u_2) = \delay(u_1v_1,u_2v_2).$};
  \tikzstyle{every node}=[font=\scriptsize]
  \path (A) edge [loop above] node {\trans{v}{v_1}} (A);
  \path (A) edge [loop left] node {\trans{u}{u_1}} (A);
  \path (A) edge node {\trans{u}{u_2}} (B);
  \path (B) edge [loop above] node {\trans{v}{v_2}} (B);
\end{tikzpicture}
\end{center}
\vspace{-2mm}
\noindent $\mathcal{T}$ satisfies the weak twinning property (WTP) if for any
two states $q_1,q_2$ of $\mathcal{T}$, $q_1$ is weakly twinned to
$q_2$. $\mathcal{T}$ is weakly twinned if it satisfies the
WTP. 
\end{definition}

\begin{remark}\label{rm:TPvsWTP} 
A transducer satisfying the twinning property also satisfies the weak
twinning property. Indeed, suppose that $\tra$ satisfies the twinning
property. We show that in any pattern depicted in
Definition~\ref{def:wtp}, we immediately get
$\delay(u_1,u_2)=\delay(u_1v_1,u_2v_2)$. Indeed, since $\tra$ is trim, 
there exist words $(x,x')\in \Sigma^*\times \Gamma^*$ such
that $q_0\xrightarrow{x\mid x'} q_1$, where $q_0$ is the
initial state of $\tra$. Then, we have:
\vspace{-2mm}
$$
\begin{array}{lllllllll}
  q_0 & \xrightarrow{xu\mid x'u_1} & q_1 & \xrightarrow{v\mid
                                                v_1} q_1 & \qquad &
  q_0 & \xrightarrow{xu\mid x'u_2} & q_2 & \xrightarrow{v\mid v_2} q_2
\end{array}
$$
Since $\tra$ satisfies the twinning property, then 
$\delay(x'u_1,x'u_2) = \delay(x'u_1v_1,x'u_2v_2)$, and therefore
$\delay(u_1,u_2) = \delay(u_1v_1,u_2v_2)$. 
\end{remark}

\vspace{-6mm} 
\subsection{Main Result}
\vspace{-2mm}

We show that the weak twinning property characterises the
transducers that define multi-sequential relations, and that it is
decidable in polynomial time. 

\vspace{-1mm}
\begin{theorem}[Main Result]\label{thm:main}
    Let $\tra$ be a trim transducer. 
\vspace{-1mm}
\begin{enumerate}
\item\label{eqWTP} $\tra$ is weakly twinned iff $\inter{\tra}$ is
  multi-sequential. 
\item\label{decWTP} It is decidable in \textsc{PTime} whether a trim
  transducer is weakly twinned. 
\end{enumerate}
\end{theorem}

Deciding the WTP is done with a reversal-bounded counter machine,
whose emptiness is known to be decidable in \textsc{PTime} \cite{JACM::Ibarra1978}
(see Appendix). The proof of Theorem
\ref{thm:main}.\ref{eqWTP} is done via two lemmas, Lemma
\ref{lem:decomposition} and \ref{lem:necessary}, that are shown in the
rest of this paper. An immediate consequence of this theorem and the
fact that any transducer can be trimmed
in polynomial time, is the following corollary:
\vspace{-1mm}
\begin{corollary}
    It is decidable in \textsc{PTime} whether a transducer defines a
    multi-sequential relation. 
\end{corollary}

\begin{remark}\label{rm:realtime} Theorem \ref{thm:main} is also true
    when $\tra$ is a rational transducer. Indeed, if $\tra$ is weakly twinned,
    then there is no loop of the form $q\xrightarrow{\epsilon|w}q$, otherwise by taking $q_1 = q_2 = q$,
    $u=v=u_1=v_1=\epsilon$ and $u_2 = v_2 = w$ in the definition of
    the WTP, one would raise a contradiction. It is easily shown that
    $\tra$ can be transformed into an equivalent real-time transducer,
    while preserving the WTP. Conversely, if $\inter{\tra}$ is
    multi-sequential, then it is finite-valued, and therefore there is
    no loop of the form $q\xrightarrow{\epsilon|w}q$. As before, one can transform $\tra$ into a real-time
    transducer while preserving the WTP. 
\end{remark}

\begin{lemma}\label{lem:decomposition}
    Let $\tra$ be a trim transducer. If $\tra$ is weakly twinned, then
    $\inter{\tra}$ is multi-sequential.
\end{lemma}

\begin{proof}
The proof of this Lemma is the goal of Sec.~\ref{sec:decomposition}
which provides a procedure that decomposes a transducer $\tra$ into a
union of sequential transducers. This procedure generalises to
relations the determinisation procedure of functional transducers. In
particular, it is based on a subset construction extended with delays, 
and a careful analysis of the strongly connected components of
$\tra$. \eof
\end{proof}

The following lemma is a key result to prove the other direction of
Theorem \ref{thm:main}.\ref{eqWTP}. It states that the WTP is
preserved by transducer inclusion and
equivalence, and therefore, is independent from the transducer that
realises the relation. 

\begin{lemma}\label{lem:preserve}
    Let $\tra_1,\tra_2$ be two trim transducers. 
\vspace{-1mm}
    \begin{enumerate}
        \item If $\inter{\tra_1} \subseteq \inter{\tra_2}$ and $\tra_2$
          is weakly twinned, then $\tra_1$ is weakly twinned. 
        \item If $\inter{\tra_1}=\inter{\tra_2}$, then $\tra_1$ is
          weakly twinned iff $\tra_2$ is weakly twinned.
    \end{enumerate}
\end{lemma}

\vspace{-2mm}
\begin{proof}
Clearly, $2$ is a consequence of $1$. Let us prove $1$ by
contradiction. Suppose that $\tra_1$ is \emph{not} weakly twinned. 
\vspace{-1mm}
Hence, it can be shown that there exist two 
states $q_1$ and $q_2$
of $\tra_1$, an accepting run $q_0 \xrightarrow{x|x'}_{\tra_1} q_1
\xrightarrow{u|u_2}_{\tra_1} q_2 \xrightarrow{y|y'}_{\tra_1} q_f$ and
loops $q_1 \xrightarrow{u|u_1}_{\tra_1} q_1
\xrightarrow{v|v_1}_{\tra_1} q_1$ and $q_2
\xrightarrow{v|v_2}_{\tra_1} q_2$ such that for every $n \in
\mathbb{N}$, $|\delay(u_1v_1^n,u_2v_2^n)| \geq n$ 
(see Lemma \ref{choice_wtp} in appendix).

Since $\tra_2$ is weakly twinned, by Lemma \ref{lem:decomposition},
there exist sequential transducers $\mathcal{D}_1, \ldots,
\mathcal{D}_k$ such that $\inter{\tra_2} = \bigcup_{i = 1}^k \inter{\mathcal{D}_i}$.
Let $M = max \{ |u|,|v|,|y| \}$.
Let $b$ be the maximal $M_{\mathcal{D}_i}$, $1 \leq i \leq k$.
Let $r = 4b(k+1)(M+1)$.
For every $1 \leq i \leq k+1$, consider the pair $(w_i,w'_i)$ in $\inter{\tra_1}$ defined by
\vspace{-1mm}
$$\begin{array}{lllllllllll}
w_i & = & xv^{r^{k}}uv^{r^{k-1}} \ldots uv^{r^{i+1}}uv^{r^{i}}y,\qquad
w'_i & = & x'v_1^{r^{k}}u_1v_1^{r^{k-1}} \ldots u_1v_1^{r^{i+1}}u_2v_2^{r^{i}}y'.
\vspace{-1mm}
\end{array}
$$
Since $\inter{\tra_1}\subseteq \inter{\tra_2}$, and
$(w_i,w'_i)\in\inter{\tra_1}$, we have $(w_i,w'_i)\in\inter{\tra_2}$, hence there exists $1 \leq j \leq k$ such that $(w_i,w'_i)\in\inter{\mathcal{D}_j}$.
As there are $k$ different $\mathcal{D}_j$ and $k+1$ pairs
$(w_i,w_i')$, there exist $k \geq i_1 > i_2 \geq 0$ such that
$(w_{i_1},w_{i_1}'), (w_{i_2},w_{i_2}')\in\inter{\mathcal{D}_j}$.
By Proposition \ref{bounded_variation}, $|\Delta(w_{i_1}',w_{i_2}')| \leq b (|\Delta(w_{i_1},w_{i_2})| + 2)$.
However,
\vspace{-1mm}
$$
\begin{scriptsize}
\begin{array}{lcl@{$\quad\qquad$} lll}
& & b(|\Delta(w_{i_1},w_{i_2})| + 2) & & |\Delta(w'_{i_1},w'_{i_2})|\\ 
& = &  b|y^{-1}uv^{r^{i_1-1}} \ldots uv^{r^{i_2}}y| + 2b  & = & |(u_2v_2^{r^{i_1}}y')^{-1}u_1v_1^{r^{i_1}} \!\!\!\!\!\ldots u_1v_1^{r^{i_2 + 1}} u_2v_2^{r^{i_2}}y'|\\
& \leq & b|uv^{r^{i_1-1}} \ldots uv^{}y| + 2b & \geq & |(u_2v_2^{r^{i_1}})^{-1}u_1v_1^{r^{i_1}}| \\

& = & i_1b|u| + b|y| +  \frac{r^{i_1}-1}{r-1}b |v| + 2b& & {-}
                                                           |u_1v_1^{r^{i_1{-}1}}
                                                           \!\!\!\!\!\!\!\!\ldots
                                                           u_1v_1^{r^{i_2
                                                           + 1}}
                                                           u_2v_2^{r^{i_2}}y'| \\

& \leq &   b(k{+}1)(M{+}1) {+} 2bM(r^{i_1{-}1}{-}1)& \geq & r^{i_1} - b|y| - b|uv^{r^{i_1-1}} \ldots uvy|\\

& \leq & \frac{r}{4} + \frac{r^{i_1}-1}{4}\ <\  \frac{r^{i_1}}{2}  & >
                                       & r^{i_1} - \frac{r^{i_1}}{2}\
                                         =\ \frac{r^{i_1}}{2}
\end{array}
\end{scriptsize}
$$
which is a contradiction.\eof
\end{proof}

\begin{lemma}\label{lem:necessary}
    Let $\tra$ be a trim transducer. If $\inter{T}$ is
    multi-sequential, then $T$ is weakly twinned. 
\end{lemma}

\begin{proof}
    If $\inter{\tra}$ is multi-sequential, then $\tra$ is equivalent to a
    transducer $\tra'$ given as a union of $k$ sequential transducers $\mathcal{D}_i$ for some $k\geq
    0$ with disjoint sets of states. Clearly, if each $\mathcal{D}_i$ is weakly twinned, then 
    so is $\tra'$. Since the $\mathcal{D}_i$ are sequential, they
    satisfy the twinning property, and therefore the weak twinning
    property by Remark \ref{rm:TPvsWTP}. Hence, $\tra'$ is weakly
    twinned. By Lemma \ref{lem:preserve} and since $\tra'$ and $\tra$
    are equivalent, $\tra$ is also weakly twinned. \eof
\end{proof}

The following result implies that,  in order to show that a rational relation is not
multi-sequential, it suffices to exhibit a function contained in that
relation, which is not multi-sequential. 

\begin{corollary}
    Let $R$ be a rational relation, and $f$ a rational function such that
    $f\subseteq R$ and $f$ is not multi-sequential. Then $R$ is not
    multi-sequential. 
\end{corollary}

\begin{proof}
    We assume that $R$ and $f$ are defined by transducers
    $\tra$ and $\tra_f$. The result still holds for 
    rational transducers, for the same reasons as the one explained in
    Remark \ref{rm:realtime}. Since $f$ is
    not multi-sequential, by Theorem \ref{thm:main}.\ref{eqWTP},
    $\tra_f$ is not weakly twinned. Since $f\subseteq R$, by Lemma
    \ref{lem:preserve} it implies that $\tra$ is not weakly twinned, and hence not
    multi-sequential, again by Theorem
    \ref{thm:main}.\ref{eqWTP}. \eof
\end{proof}



\section{Decomposition Procedure}\label{sec:decomposition}

In this section, we show how to decompose a transducer into a union of
sequential transducers, via a series of constructions, whenever the weak twinning property is
satisfied. For simplicity, we sometimes consider \emph{multi-transducers},
i.e. transducers such that the function $f_T$ maps any final state to
a finite set of output words. Let $\tra = (Q,E,I,T,f_T)$ be a transducer.
Let $\sim\subseteq Q^2$ defined by $q_1 \sim q_2$ if $q_1$ and $q_2$
are strongly connected, i.e. if there exist a run from $q_1$ to $q_2$ and a run from $q_2$ to $q_1$.
The equivalence classes of $\sim$ are called the strongly connected
components (SCC) of $\tra$. An edge of $\tra$ is called \emph{transient} if its source and 
target are in distinct SCC, or equivalently, if there exist no run from its target to its source.
The condensation of $\tra$ is the directed acyclic graph $\Psi (\tra)$ whose vertices are the SCC of $\tra$ and whose edges are the transient edges of $\tra$.
A transducer is called \emph{separable} if it has a single initial
state and any two edges of same source and same input symbol are
transient.

\paragraph{\textbf{Split}} Let $\tra = (Q,E,I,T,f_T)$ be a transducer. 
Let $P$ be the paths of the condensation $\Psi(\tra)$ starting in an SCC containing an initial state. Note that $P$ is finite as $\Psi (\tra)$ is a DAG.
For each path $p \in P$, let $\tra_p$ be the subtransducer of $\tra$
obtained by removing all the transient edges of $\tra$ but the ones
occurring in $p$. The split of $\tra$ is the transducer
$\textsf{split}(\tra) =\bigcup_{p \in P}\tra_p$. Clearly,

\begin{lemma}\label{split_eq}
The transducer $\textsf{split}(\tra)$ is equivalent to $\tra$,
i.e. $\inter{\textsf{split}(\tra)}=\inter{\tra}$. 
\end{lemma}


If $\tra$ is separable, then $\textsf{split}(\tra)$ is a decomposition
of $\tra$ into sequential transducers. Since any multi-transducer can
be transformed into an equivalent union of transducers over the same
underlying automaton while preserving separability, we get the
following result (fully proved in Appendix):

\begin{lemma}\label{sepmultiseq}
Let $\tra$ be a separable multi-transducer with a single initial state.
Then $\inter{\tra}$ is multi-sequential. 
\end{lemma}

\paragraph{\textbf{Determinisation}} We recall the determinisation
procedure for transducers, for instance presented in 
\cite{DBLP:journals/tcs/BealC02}. It extends the subset construction
with delays between output words, and outputs the longest common
prefix of all the output words produced on transitions on the same input
symbol, and keep the remaining suffixes (delays) in the states. 
Precisely, let $\tra = (Q,E,I,T,f_T)$ be a trim transducer. For every $U \in \mathcal{P}_f(Q \times \Gamma^*)$, for every $\sigma \in \Sigma$, let 
$$\begin{array}{ll}
R_{U,\sigma} & = \{(q,w) \in Q \times \Gamma^* | \exists (p,u) \in U, \exists (p,\sigma|v,q) \in E \textup{ s.t. } w = uv\},\\
w_{U,\sigma} & \textup{be the largest common prefix of the words $\{ w | \exists q \in Q \textup{ s.t. } (q,w) \in R_{U,\sigma}$} \} ,\\
P_{U,\sigma} & = \{(q,w) | (q,w_{U,\sigma}w) \in R_{U,\sigma} \}.
\end{array}$$
Let $\bar{\mathcal{D}}(\tra)$ be the infinite-state multi-transducer
over the set of states $\mathcal{P}_f(Q \times \Gamma^*)$, with set of
edges $\{ (U,\sigma|w_{U,\sigma},P_{U,\sigma}) | U \in \mathcal{P}_f(Q
\times \Gamma^*), \sigma \in \Sigma \}$,  initial state $U_0 = I\times
\{\epsilon\}$, set of final states $\{ U \in \mathcal{P}_f(Q \times \Gamma^*) | U \cap (T \times \Gamma^*) \neq \emptyset \}$, and final output relation that maps each final state $U$ to $\{wf_T(q) | q \in T \textup{ and } (q,w) \in U \}$.
Note that $\bar{\mathcal{D}}(\tra)$ has a deterministic (potentially
infinite) input-automaton.

The determinisation of $\tra$, written $\mathcal{D}(\tra)$, is the trim part of $\bar{\mathcal{D}}(\tra)$.
The transducer $\mathcal{D}(\tra)$ is equivalent to $\tra$ (see corollary \ref{det_equ}, appendix).
It is
well-known that $\mathcal{D}(\tra)$ is a (finite) sequential
transducer iff $\tra$ satisfies the twinning property.

Fig.~\ref{exmultnseq} depicts a transducer that satisfies the weak twinning
property, but not the twinning property. As a consequence,
$\mathcal{D}(\tra)$ is infinite (a part of $\mathcal{D}(\tra)$ can be
seen on Fig.~\ref{exmultnseq-det}). The non satisfaction of the
twinning property is witnessed by the two runs $q_0\xrightarrow{aaaa|abaa} q_4 \xrightarrow{a|a} q_4$ and $q_0\xrightarrow{aaaa|baaa} q_4 \xrightarrow{a|a} q_4$. Note that these runs do not harm the weak
twinning property. The idea of the next construction,
called the weak determinisation, is to keep some, well-chosen, non-deterministic
transitions, and reset the determinisation whenever
it definitively leaves an SCC (the SCC $\{q_0,q_1,q_2\}$ in this
example). We explain this procedure when there is a single initial
state, as any transducer can be easily transformed into a finite
union of transducers with single initial states.

\begin{figure}[!ht]
\centering
\subfigure[$\tra$\label{exmultnseq}]{
\begin{tikzpicture}[->,>=stealth',auto,node distance=2cm,thick,scale=0.7,every node/.style={scale=0.7}]
  \tikzstyle{every state}=[fill=yellow!30,text=black]
  \tikzstyle{initial}=[initial by arrow, initial where=left, initial text=]
  \tikzstyle{accepting}=[accepting by arrow, accepting where=right, accepting text=]

  \node[initial,state] (A)  {$q_0$};
  \node[state] [right of=A] (B)  {$q_1$};
  \node[state] [below of=A] (C)  {$q_2$};
  \node[state] [right of=C] (D)  {$q_3$};
  \node[accepting,state] [right of=D] (E)  {$q_4$};
  \path (A) edge node {\trans{a}{a}} (B);
  \path (A) edge [bend right,left] node {\trans{a}{b}} (C);
  \path (B) edge node {\trans{a}{b}} (C);
  \path (C) edge [bend right] node  {\trans{b}{\epsilon}} (A);
  \path (C) edge [below] node {\trans{a}{a}} (D);
  \path (D) edge [below] node {\trans{a}{a}} (E);
  \path (E) edge [loop above] node {\trans{a}{a}} (E);
\end{tikzpicture}
}
\subfigure[$\mathcal{D}(\tra)$\label{exmultnseq-det}]{
\begin{tikzpicture}[->,>=stealth',auto,node distance=2cm,thick,scale=0.7,every node/.style={scale=0.7}]
  \tikzstyle{every state}=[fill=yellow!30,text=black, shape = rectangle,rounded corners,font=\scriptsize]
  \tikzstyle{initial}=[initial by arrow, initial where=left, initial text=]
  \tikzstyle{accepting}=[accepting by arrow, accepting where=right, accepting text=$\cdots$]

  \node[initial,state] (A)  {$(q_0, \epsilon)$};
  \node[state] [right of=A] (B)  {$\begin{array}{l}(q_1, a)\\ (q_2, b)\end{array}$};
  \node[state] [below of=A] (C)  {$\begin{array}{lll}(q_2, ab)\\ (q_3, ba)\end{array}$};
  \node[state] [right of=C] (D)  {$\begin{array}{lll}(q_3, aba)\\ (q_4, baa)\end{array}$};
  \node[state,accepting] [right of=D] (E)  {$\begin{array}{lll} (q_4, abaa)\\ (q_4, baaa)\end{array}$};
  \path (A) edge [bend right] node {\trans{a}{\epsilon}} (B);
  \path (B) edge node {\trans{a}{\epsilon}} (C);
  \path (C) edge [below] node {\trans{a}{\epsilon}} (D);
  \path (B) edge [bend right,above] node {\trans{b}{b}} (A);
  \path (C) edge node {\trans{b}{ab}} (A);
  \path (D) edge [below] node {\trans{a}{\epsilon}} (E);
\end{tikzpicture}
}
\vspace{-4mm}
\caption{\label{fig:det1} Non determinisable transducer.}
\vspace{-6mm}
\end{figure}
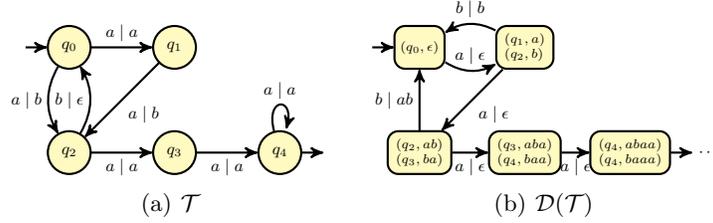

\paragraph{\textbf{Weak determinisation}} Let $\tra = (Q,E,I,T,f_T)$ be a trim transducer with a single initial state.
For every $U \in \mathcal{P}_f(Q \times \Gamma^*)$, let the rank $n_{U}$
of $U$ be the set containing all the SCC of $\tra$ accessible from the
states $q\in U$. The multi-transducer $\bar{\mathcal{W}}(\tra)$ is obtained from $\bar{\mathcal{D}}(\tra)$ by splitting the edges that do not preserve the rank, as follows.
If $(U,u|v,U')$ is an edge of $\bar{\mathcal{D}}(\tra)$ such that
$n_{U'}$ is strictly included in $n_{U}$, it is removed, and replaced
by the set of edges $\{ (U,u|vw,\{(q,\epsilon)\}) | (q,w) \in
U'\}$. It is easily shown that any pair of distinct edges of the form 
$U\xrightarrow{a|v_1} U_1$ and $U\xrightarrow{a|v_2} U_2$ in $\bar{\mathcal{W}}(\tra)$ have necessarily been
created by this transformation (because without this transformation
everything stays input-deterministic). Therefore, since the rank
strictly decreases ($n_{U_2}\subsetneq n_U$ and $n_{U_1}\subsetneq
n_U$) and can never increase in $\bar{\mathcal{W}}(\tra)$, there is no
run from $U_2$ to $U$, nor from $U_1$ to $U$ in
$\bar{\mathcal{W}}(\tra)$, and the two edges are transient. As a
consequence, 

\begin{lemma}\label{lem:separable}
The infinite transducer $\bar{\mathcal{W}}(\tra)$ is separable.
\end{lemma}

The weak determinisation of $\tra$, written $\mathcal{W}(\tra)$, is the trim part of $\bar{\mathcal{W}}(\tra)$.

\begin{proposition}\label{weakdet}
$\mathcal{W}(\tra)$ and $\tra$ are equivalent. Moreover, 
if $\tra$ is weakly twinned, $\mathcal{W}(\tra)$ is finite, and it
is a multi-transducer.
\end{proposition}

The main idea is to prove that, as long as the weak twinning property
is satisfied, the length of the words present in the states of
$\mathcal{W}(\tra)$ can be bounded. The proof can be found in
Appendix.

\begin{example}
Let us illustrate the weak determinisation on the transducer $\tra$ of
Fig.~\ref{exmultnseq}. Consider the determinisation
$\mathcal{D}(\tra)$ of $\tra$ of Fig.~\ref{exmultnseq-det}.
When it is in state $U_1 = \{(q_2,ab),(q_3,ba)\}$, on input $a$, it moves to
state $U_2 = \{(q_3,aba),(q_4,baa)\}$, definitely leaving the SCC
$\{q_0,q_1,q_2\}$ of $\tra$ (the rank $n_{U_2}$ of $U_2$ is
strictly included in the rank $n_{U_1}$ of $U_1$). As a result, this
transition is removed from $\bar{\mathcal{D}}(\tra)$, and replaced by the
transitions $U_2\xrightarrow{a|aba} \{ (q_3,\epsilon)\}$ and
$U_2\xrightarrow{a|baa} \{ (q_4,\epsilon)\}$. The resulting transducer
$\mathcal{W}_\tra$ is depicted on Fig.~\ref{fig:wdet1} (where the new
transitions are dotted). Fig.~\ref{fig:split_wdet1} shows how the
latter transducer is split into a union.
\end{example}


\paragraph{Proof of Lemma \ref{lem:decomposition}} We can finally prove that the every weakly twinned transducer is multi-sequential.
Let $\tra = (Q,E,I,T,f_T)$ be a weakly twinned transducer.
Then $\tra$ is equivalent to $\bigcup_{i \in I}\tra_i$, where $\tra_i$ is the transducer obtained by keeping only $i$ as initial state.
Given $i \in I$, as we just saw, $\mathcal{W}(\tra_i)$ is a transducer equivalent to $\tra_i$.
Moreover, as $\bar{\mathcal{W}}(\tra_i)$ is separable, so is $\mathcal{W}(\tra_i)$,hence, by Lemma $\ref{sepmultiseq}$, it is multi-sequential.
The desired result follows.

\begin{figure}[!ht]
\subfigure[$\mathcal{W}(\tra)$\label{fig:wdet1}]{
\begin{tikzpicture}[->,>=stealth',auto,node distance=2cm,thick,scale=0.7,every node/.style={scale=0.7}]
  \tikzstyle{every state}=[fill=yellow!30,text=black, shape = rectangle,rounded corners,font=\scriptsize]
  \tikzstyle{initial}=[initial by arrow, initial where=left, initial text=]
  \tikzstyle{accepting}=[accepting by arrow, accepting where=right, accepting text=]

  \node[initial,state] (A) at (0,0) {$(q_0, \epsilon)$};
  \node[state] (B) at (2.2,0) {$\begin{array}{l}(q_1, a)\\ (q_2, b)\end{array}$};
  \node[state] [below of=A] (C)  {$\begin{array}{lll}(q_2, ab)\\ (q_3, ba)\end{array}$};
  \node[state] (D) at (2.2,-2) {$\begin{array}{lll}(q_3, \epsilon)\end{array}$};
  \node[state,accepting] [right of=D] (E)  {$\begin{array}{lll} (q_4, \epsilon)\end{array}$};
  \path (A) edge [bend right] node {\trans{a}{\epsilon}} (B);
  \path (B) edge node {\trans{a}{\epsilon}} (C);
  \path (C) edge [densely dotted] node {\trans{a}{aba}} (D);
  \path (C) edge [bend right,below,densely dotted] node {\trans{a}{baa}} (E);
  \path (B) edge [bend right,above] node {\trans{b}{b}} (A);
  \path (C) edge node {\trans{b}{ab}} (A);
  \path (D) edge node {\trans{a}{a}} (E);
  \path (E) edge [loop above] node {\trans{a}{a}} (E);
\end{tikzpicture}
}
\subfigure[$\textsf{trim}(\textsf{split}(\mathcal{W}(\tra)))$\label{fig:split_wdet1}]{
\begin{tikzpicture}[->,>=stealth',auto,node distance=2cm,thick,scale=0.7,every node/.style={scale=0.7}]
  \tikzstyle{every state}=[fill=yellow!30,text=black, shape = rectangle,rounded corners]
  \tikzstyle{initial}=[initial by arrow, initial where=left, initial text=]
  \tikzstyle{accepting}=[accepting by arrow, accepting where=right, accepting text=]

  \node[initial,state] (A) at (0,1) {$(q_0, \epsilon)$};
  \node[state] (B) at (2.2,1) {$\begin{array}{l}(q_1, a)\\ (q_2, b)\end{array}$};
  \node[state] (C) at (0,-1) {$\begin{array}{lll}(q_2, ab)\\ (q_3, ba)\end{array}$};
  \node[state] (D) at (2.2,-1) {$\begin{array}{lll}(q_3, \epsilon)\end{array}$};
  \node[state,accepting] (E) at (4.2,-1)  {$\begin{array}{lll} (q_4, \epsilon)\end{array}$};
  \path (A) edge [bend right] node {\trans{a}{\epsilon}} (B);
  \path (B) edge node {\trans{a}{\epsilon}} (C);
  \path (C) edge [densely dotted] node [below] {\trans{a}{aba}} (D);
  \path (B) edge [above,bend right] node {\trans{b}{b}} (A);
  \path (C) edge node {\trans{b}{ab}} (A);
  \path (D) edge node [below] {\trans{a}{a}} (E);
  \path (E) edge [loop above] node {\trans{a}{a}} (E);

  \node[initial,state] (A') at (6.2,1)  {$(q_0, \epsilon)$};
  \node[state] (B') at (8.7,1) {$\begin{array}{l}(q_1, a)\\ (q_2, b)\end{array}$};
  \node[state] (C') at (6.2,-1)  {$\begin{array}{lll}(q_2, ab)\\ (q_3, ba)\end{array}$};
  \node[state,accepting] (E') at (8.7,-1) {$\begin{array}{lll} (q_4, \epsilon)\end{array}$};
  \path (A') edge [bend right] node {\trans{a}{\epsilon}} (B');
  \path (B') edge node {\trans{a}{\epsilon}} (C');
  \path (C') edge [densely dotted] node [below] {\trans{a}{baa}} (E');
  \path (B') edge [above,bend right] node {\trans{b}{b}} (A');
  \path (C') edge node {\trans{b}{ab}} (A');
  \path (E') edge [loop above] node {\trans{a}{a}} (E');
\end{tikzpicture}
}
\vspace{-4mm}
\caption{Weak determinisation and split}
\vspace{-8mm}
\end{figure}
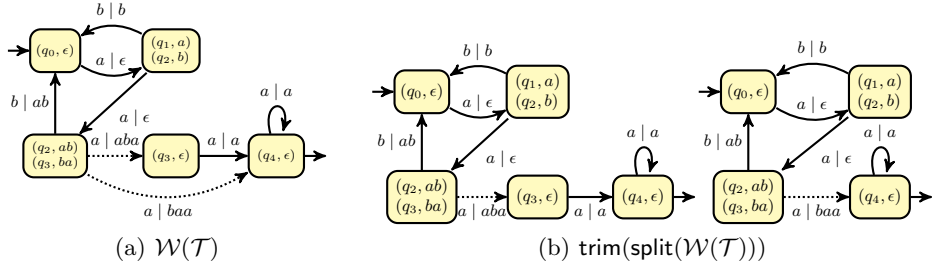


\section{Application to Multi-Output Streamability Problem}

Sequential functions have the advantage of being efficiently
computable. They are exactly the word functions that can be evaluated
with constant memory in a sequential, left-to-right, manner. This
computability notion have been defined formally in
\cite{DBLP:conf/fsttcs/FiliotGRS11} with the model of Turing
transducers. Informally, a Turing transducer has three tapes: a
read-only left-to-right input tape, a working tape, and a write-only
left-to-right output tape. The amount of memory is measured only on
the working tape. For any sequential function $f$, there exists a
Turing transducer $M$ and constant $K\in\mathbb{N}$ such that for all words $u\in
dom(f)$, $f(u)$ can be computed by $M$ while using at most $K$ cells of
the working tape. This model is a streaming model in the sense that 
the input tape is left-to-right, and therefore one can think of
receiving the input word $u$ as a stream. The converse also holds
true: any word function computable with constant memory by a Turing
transducer is sequential. Therefore, the following problem, called the
streamability problem, is decidable in \textsc{PTime}, based on the
twinning property: given a functional transducer, does it define a
function that can be evaluated with constant memory? In this section,
we establish a similar result for relations.

We extend the model of Turing transducer to a model for computing
relations. We rather explain this model in words, avoiding a tedious
and technical definition of intuitive concepts. This model can be thought of as a streaming model where the input
word $u$ is a stream, and the outputs words are produced on-the-fly, while
processing $u$, and sent through different channels (represented as
output tapes in the model). More precisely, let $k\in\mathbb{N}$. 
A \emph{$k$-output Turing
transducer} $M$ is a deterministic Turing machine with $(1)$ a read-only
left-to-right input tape, $(2)$ a two-way working tape, and $(3)$ $k$
write-only left-to-right output tapes without stay
transitions (therefore a cell cannot be rewritten). Additionally, the
machine can disable/enable some output tapes. A \emph{tape
  configuration} of such a machine is a tuple $(u,v,W)$ where $u$ is
the content of the input tape, $v$ is the content of the working tape,
and $W$ is the set of contents of the enabled output tapes. By content we mean the sequence of
symbols to the first blank symbol.

To define the class of constant memory computable relations, we will
allow some preprocessing of the input stream, i.e., the computation of 
a constant amount of information that can be then exploited by the
machine to compute the output words. In the setting of sequential
function, this information is implicitly present in the definition of
constant memory computability: it is assumed that the input stream belongs to the domain
of the function, otherwise it could start producing some output word
and realise later on that the word was not in the domain. The
information of being in the domain or not (a 0/1 bit) is typically some
preprocessing computation that can be performed on the input word, by
some other application. For instance, if the stream is generated by
another application, this application could also send the information
on whether the stream belong to the domain of the function or not.

We come to the definition of constant memory computability for
relations. Let $R\subseteq \Sigma^*\times
\Gamma^*$. We say that $R$ is \emph{constant memory computable} if
there exist two constants $k, K\in\mathbb{N}$, a
computable function $f : \Sigma^* \rightarrow \{0,1\}^*$ such that for
all $u\in\Sigma^*$, $|f(u)|\leq k$, and a $k$-output Turing transducer $M$ over the alphabet $\Sigma \cup
\Gamma \cup \{0,1,\#,\$\}$ (where $\$$ stand for a blank symbol) such
that for all $u\in\Sigma^*$, on initial tape
configuration $(f(u)\#u, \epsilon, \varnothing)$ (all output tapes are
initially disabled), the machine uses at most $K$ cells of the working tape
along its computation,  and halts in  a tape configuration $(f(u)\#u,
\alpha, W)$ such that  $W = R(u)$.

\begin{theorem}\label{thm:streaming}
    Let $R\subseteq \Sigma^*\times \Gamma^*$. Then $R$ is constant
    memory computable iff it is multi-sequential. Moreover, if $R$ is given by a transducer, it is decidable in
          \textsc{PTime} whether it is constant memory computable. 
\end{theorem}



\bibliographystyle{abbrv}
\bibliography{biblio}

\newpage

\appendix
\noindent{\bf \Large Appendix}
\medskip 

\subsection*{Proofs of section 2}

\begin{lemma}\label{choice_wtp0}
Let $\tra = (Q,E,I,T,f_T)$ be a transducer that does not satisfy the weak twining property.
Then there exists a run 
$$q_1 \xrightarrow{u\mid u_1} q_1 \xrightarrow{v\mid v_1} q_1 \xrightarrow{u\mid u_2} q_2 \xrightarrow{v\mid v_2} q_2$$
such that either $|v_1| \neq |v_2|$, or $v_1,v_2 \neq \epsilon$ and there is a mismatch between $u_1$ and $u_2$.
\end{lemma}

\begin{proof}
Suppose that $\tra$ is not weakly twinned.
Then there exists a run 
$$q_1 \xrightarrow{u\mid u_1} q_1 \xrightarrow{v\mid v_1} q_1 \xrightarrow{u\mid u_2} q_2 \xrightarrow{v\mid v_2} q_2$$
such that $\delay(u_1,u_2) \neq \delay(u_1v_1,u_2v_2)$.
If $|v_1| \neq |v_2|$, we are done. 
If $|v_1| = |v_2|$, then they are both distinct from $\epsilon$, otherewise $\delay(u_1,u_2)$ would be equal to $\delay(u_1v_1,u_2v_2)$.
We have to consider three possibilities.\\
If there is a mismatch between $u_1$ and $u_2$, we are done.\\
If $u_1$ is a prefix of $u_2$, suppose ab absurdo that for every $n \in \mathbb{N}$, $u_1v_1^n$ is a prefix of $u_2v_2^n$.
In particular, this means that there exist $k,k' \in \mathbb{N}$ such that $u_2 = u_1v_1^kv_1'$, where $v_1'$ is a prefix of $v_1$, and $u_2v_2 = u_1v_1^{k'}v_1''$, where $v_1''$ is a prefix of $v_1$.
However, as $|v_1| = |v_2|$ by supposition, $k' = k+1$ and $|v_1'| = |v_1''|$, hence $v_1'=v_1''$.
Therefore 
$$\delay(u_1v_1,u_2v_2) = \delay(u_1v_1,u_1v_1^{k+1}v_1') = \delay(\epsilon,v_1^{k}v_1') = \delay(u_1,u_1v_1^{k}v_1') = \delay(u_1,u_2),$$
which is a contradiction.\\
Therefore there exists $n \in \mathbb{N}$ such that $u_1v_1^n$ is not a prefix of $u_2v_2^n$.
As $|u_1v_1^n| \leq |u_2v_2^n|$, there is a mismatch between $u_1v_1^n$ and $u_2v_2^n$, hence the run
$$q_1 \xrightarrow{uv^n\mid  u_1v_1^n} q_1 \xrightarrow{v\mid v_1} q_1 \xrightarrow{uv^n\mid u_2v_2^n} q_2 \xrightarrow{v\mid v_2} q_2$$
satisfies the desired conditions.\\
Finally, if $u_2$ is a prefix of $u_1$, the desired result can be proved in a similar way.
\end{proof}

\begin{lemma}\label{choice_wtp}
Let $\tra = (Q,E,I,T,f_T)$ be a transducer that does not satisfy the weak twining property.
Then there exists a run $$q_1 \xrightarrow{u\mid u_1} q_1 \xrightarrow{v\mid v_1} q_1 \xrightarrow{u\mid u_2} q_2 \xrightarrow{v\mid v_2} q_2$$
such that for every $n \in \mathbb{N}$, $|\delay(u_1v_1^n,u_2v_2^n)| \geq n$.
\end{lemma}

\begin{proof}
By the preceding lemma, there exists a run 
$$q_1 \xrightarrow{u\mid u_1} q_1 \xrightarrow{v\mid v_1} q_1 \xrightarrow{u\mid u_2} q_2 \xrightarrow{v\mid v_2} q_2$$
such that either $|v_1| \neq |v_2|$, or $v_1,v_2 \neq \epsilon$ and there is a mismatch between $u_1$ and $u_2$.\\
If $|v_1| < |v_2|$, we have to consider two possibilities.
If $|u_1| \leq |u_2|$, $|\delay(u_1v_1^n,u_2v_2^n)| \geq |u_2v_2^n| - |u_1v_1^n| \geq n$.
If $|u_1| > |u_2|$, by considering the run
$$q_1 \xrightarrow{u\mid u_1} q_1 \xrightarrow{uv^{|u_1| - |u_2|} \mid u_1v_1^{|u_1| - |u_2|}} q_1 \xrightarrow{uv^{|u_1| - |u_2|}\mid u_2v_2^{|u_1| - |u_2|}} q_2 \xrightarrow{v\mid v_2} q_2,$$
we are back in the first situation.\\
If $|v_1| > |v_2|$, we can proceed in a similar way.\\
Finally, if $u_1,u_2 \neq \epsilon$ and there is a mismatch between $u_1$ and $u_2$, the delay grows at each step by $|v_1| + |v_2| > 1$, and the result follows.
\end{proof}

\paragraph{Proof of Proposition \ref{bounded_variation}}

\begin{proof}
Let $\textsf{p}(u_1,u_2)$ denote the largest common prefix of $u_1$ and $u_2$, and let $\textsf{d}(u_1,u_2)$ and $\textsf{d}(u_2,u_1)$ be the words such that $u_1 = \textsf{p}(u_1,u_2) \textsf{d}(u_1,u_2)$ and $u_2 = \textsf{p}(u_1,u_2) \textsf{d}(u_2,u_1)$.
We use similar notations for $v_1$ and $v_2$.\\
As $(u_1,v_1)$ and $(u_2,v_2)$ are recognised by $\mathcal{D}$, there exist two accepting runs $p_1 : \bar{q} \xrightarrow{u_1|v_1'} q_1$ and $p_2 : \bar{q} \xrightarrow{u_2|v_2'} q_2$, where $v_1'f_T(q_1) = v_1$ and $v_2'f_T(q_2) = v_2$.
As the input automaton is deterministic, each run can be decomposed in two parts as follows. 
$$\begin{array}{lcl}
p_1 & : &  \bar{q} \xrightarrow{\textsf{p}(u_1,u_2)|v} q \xrightarrow{\textsf{d}(u_1,u_2)|v_1''} q_1,\\
p_2 & : & \bar{q} \xrightarrow{\textsf{p}(u_1,u_2)|v} q \xrightarrow{\textsf{d}(u_2,u_1)|v_2''} q_2.
\end{array}$$
Therefore, $v$ is a prefix of both $v_1$ and $v_2$, hence it is a prefix of $\textsf{p}(v_1,v_2)$.
Finally,
$$\begin{array}{lcl}
|\Delta(v_1,v_2)| & = & |\textsf{d}(v_1,v_2)| + |\textsf{d}(v_2,v_1)|\\
& = & |v_1| + |v_2| - 2|\textsf{p}(v_1,v_2)|\\
& = & 2 |v| + |v_1''| + |v_2''| + |f_T(q_1)| + |f_T(q_2)| - 2|\textsf{p}(v_1,v_2)|\\
& \leq & |v_1''| + |v_2''| + |f_T(q_1)| + |f_T(q_2)|\\
& \leq & M_{\mathcal{D}} |\textsf{d}(u_1,u_2)| + M_{\mathcal{D}} |\textsf{d}(u_2,u_1)| + 2 M_{\mathcal{D}}\\
& = & M_{\mathcal{D}} (|\Delta(u_1,u_2)|+2).
\end{array}$$
\end{proof}

\subsection*{Proof of Theorem \ref{thm:main}.\ref{decWTP}}

\begin{proof}
    Given a trim transducer $\tra = (Q,E,I,T,f_T)$, we construct a counter machine $M$
whose counters are reversal-bounded (they alternate between increasing
and decreasing modes only a bounded number of times), such that $L(M)$
is empty iff $\tra$ is weakly twinned. The result will follow since 
emptiness of reversal-bounded counter machine is decidable in
\textsc{PTime} \cite{JACM::Ibarra1978}. We base our construction on
checking the conditions of Lemma \ref{choice_wtp0}.

Let $\#$ be a fresh symbol that does not belong to $\Sigma$. The machine
$M$ will be the union of two machines $M_1$ and $M_2$ that
respectively accept words of the form $u\#v$ such that there exist
$q_1,q_2\in Q$ and $v_1,v_2,u_1,u_2\in\Gamma^*$ such that
$$
q_1\xrightarrow{u|u_1} q_1\xrightarrow{v|v_1}q_1\xrightarrow{u|u_2} q_2\xrightarrow{v|v_2} q_2
$$
and 
\begin{itemize}
  \item $|v_1|\neq |v_2|$ (for $M_1$)
  \item $v_1,v_2\neq \epsilon$ and there is a mismatch between $u_1$
    and $u_2$ (for $M_2$)
\end{itemize}

Let us explain how to construct $M_1$. First, to check the property
about runs, $M_1$ non-deterministically guesses $q_1$ and $q_2$, 
and will simulate two runs of $\tra$ in parallel on $u$ and $v$. While
reading $u$, $M_1$ simulates the product of $\tra$ with itself, by guessing
two runs of $\tra$ from the pair of state $(q_1,q_1)$ to the pair of
states $(q_1,q_2)$.  Similarly, it will simulate two runs of $\tra$
from the pair $(q_1,q_2)$ to the pair $(q_1,q_2)$, on $v$. By
simulating these runs, $M_1$ will non-deterministically trigger 
transitions of $\tra$, and therefore will guess output words
$u_1,u_2,v_1,v_2$. In order to check that $|v_1|\neq |v_2|$, the
machine $M_1$ uses two counters $c_1,c_2$, both initialised,
non-deterministically, to some value $\iota$ (by using a loop that 
increments both $c_1$ and $c_2$, and by non-deterministically deciding
to leave that loop). When reading $v$, if $M$ simulates two
transitions of $\tra$ of the form $q\xrightarrow{a|w} q'$ and $p\xrightarrow{a|w'} p'$ in
parallel, the machine $M_1$ decreases the counters $c_1$ and $c_2$
from respectively the value $|w|$ and $|w'|$. It accepts only if it
ends in the pair of states $(q_1,q_2)$, $c_1 = 0$, and $c_2>0$
(meaning that the guessed value $\iota$ equals $|v_1|$, and
$|v_2|>|v_1|$), or if $c_2 = 0$ and $c_1 > 0$ (meaning $|v_1|>
|v_2|=\iota$). Counter machines can only test for zero, but this can
be easily turned into more complicated tests, such as $c_1 = 0$ and
$c_2 >0$, with polynomial space.

The machine $M_2$ is constructed similarly. It can easily check that
$v_1,v_2\neq \epsilon$ (if the transitions of $\tra$ it simulates in
parallel only produce empty words, then the machine stays in a
non-accepting state). To check the existence of a mismatch between
$u_1$ and $u_2$, the machine again uses two counters $c_1,c_2$
both initialised to some value $\iota$. It will check that $\iota \leq
|u_1|,|u_2|$, and that $u_1[\iota]\neq u_2[\iota]$. Like for $M_1$, this is done by
decreasing in parallel $c_1$ and $c_2$ by the length of the respective
words produced on the simulated transitions of $\tra$ when reading
$u$. If one of the two counters reaches $0$, say $c_1$, then the
symbol $u_1[\iota]$ is stored in the state of $M_2$, and $c_2$ still
continues to be decreased, until it reaches $0$. At that point, it
suffices to check that the symbol $u_2[\iota]$ is different from the
stored symbol, and to accept only in that case.

The two machines $M_1$ and $M_2$ can be constructed in polynomial time
from $\tra$, and therefore have polynomial size. Their counters are
reversal-bounded: there is only one change of polarity. Therefore, the
emptiness of $M_1$ and $M_2$ can be checked in polynomial time (in the
size of $\tra$). \eof
\end{proof}




\subsection*{Proofs of section 3}

\paragraph{Proof of lemma \ref{sepmultiseq}}

\begin{proof}
First, we prove the result for transducers.
Let $\tra = (Q,E,I,T,f_T)$ be a separable transducer.
We shall show that the trim part of $\textsf{split}(\tra)$ is a decomposition of $\tra$ into sequential transducers.
Using the same notations as in the definition, we have
$\textsf{split}(\tra) = \bigcup_{p \in P}\tra_p.$
Then for every $p$ in $P$, the input automaton $\mathcal{A}_p$ of $\tra_p$ has a single initial state.
Suppose ab absurdo that $\mathcal{A}_p$ admits two distinct edges of same source and same label.
By supposition, both are transient edges of $\tra$, hence, by definition of $\tra_p$, they are part of the path $p$.
However, as they share the same source, the first one encountered by $p$ is not transient, which is a contradiction.
Therefore $\mathcal{A}_p$ is deterministic, hence $\tra_p$ is sequential, which proves the desired result.\\
The extension to multi-transducers uses the fact that any multi-transducer can be transformed in an equivalent finite union of transducers over the same underlying automaton, preserving separability.
\end{proof}

\begin{lemma}\label{run_equ}
Given $U \in \mathcal{P}_f(Q,\Gamma^*)$ and $v \in \Gamma^*$, let
$$\begin{array}{ll}
R_{U,v} & = \{(q,w) \in Q \times \Gamma^* | \exists (p,u) \in U,v'\in \Gamma^* \textup{ s.t. } p \xrightarrow{v|v'}_{\tra} q \textup{ and } w = uv'\},\\
w_{U,v} & \textup{be the largest common prefix of the words $\{ w | \exists q \in Q \textup{ s.t. } (q,w) \in R_{U,v}$} \} ,\\
P_{U,v} & = \{(q,w) | (q,w_{U,v}w) \in R_{U,v} \}.
\end{array}$$
Then $r : U \xrightarrow{v|w_{U,v}} P_{U,v}$ is the only run of source $U$ and input $v$ in $\bar{\mathcal{D}}(\tra)$.
\end{lemma}

\begin{proof}
This is proved by induction on the length of $v$.\\
If $v = \sigma \in \Sigma$, this follows immediately from the definition of the edges of $\bar{\bar{\mathcal{D}}}(\tra)$.\\
Otherwhise, $v = v_0\sigma$ for some non-empty word $v_0$ and some letter $\sigma \in \Sigma$.
Suppose that the result is true for $v_0$.
Then the only run of source $U$ and input $v_0$ in $\bar{\mathcal{D}}(\tra)$ is 
$$U \xrightarrow{v_0|w_{U,v_0}} P_{U,v_0}.$$
By definition of the edges of $\bar{\mathcal{D}}(\tra)$, the only edge of source $P_{U,v_0}$ and input $\sigma$ is $$P_{U,v_0} \xrightarrow{\sigma|w_{P_{U,v_0},\sigma}} P_{P_{U,v_0},\sigma}.$$
Hence the only run of source $U$ and input $v$ in $\bar{\mathcal{D}}(\tra)$ is 
$$U \xrightarrow{v|w_{U,v_0}w_{P_{U,v_0},\sigma}} P_{P_{U,v_0},\sigma}.$$
However, by definition, $R_{U,v} = \{ (q,w_{U,v_0}v') | (q,v') \in R_{P_{U,v_0},\sigma} \}$,
Therefore $w_{U,v} = w_{U,v_0}w_{P_{U,v_0},\sigma}$ and $P_{U,v} = P_{P_{U,v_0},\sigma}$, which proves the desired result.
\end{proof}

\begin{corollary}\label{det_equ}
The transducer $\tra$ is equivalent to $\mathcal{D}(\tra)$.
\end{corollary}

\begin{proof}
Let $(u,v) \in \Sigma^* \times \Gamma^*$.
Then
$$\begin{array}{lcl}
(u,v) \in \inter{\tra} & \Leftrightarrow & \textup{there exists an accepting run } q_0 \xrightarrow{u|w}_{\tra} q \textup{ s.t. } v = wf_T(q)\\
& \Leftrightarrow & \exists q \in T, \exists w \in \Gamma^* \textup{ s.t. } (q,w) \in R_{U_0,u} \textup{ and }v = wf_T(q)\\
& \Leftrightarrow & \exists q \in T, \exists w' \in \Gamma^* \textup{ s.t. } (q,w) \in P_{U_0,u} \textup{ and }v = w_{U_0,u}w'f_T(q)\\
& \Leftrightarrow & (u,v) \in \inter{\mathcal{D}(\tra)}\\
\end{array}$$
where the last equivalence follows from lemma \ref{run_equ}.
\end{proof}

\begin{corollary}\label{wdet_equ}
The transducer $\tra$ is equivalent to $\mathcal{W}(\tra)$.
\end{corollary}

\begin{proof}
Using the definition of $\mathcal{W}(\tra)$, and the fact that, by lemma \ref{run_equ}, for every $U \in \mathcal{T}(Q \times \Gamma^*)$, for every $v \in \Sigma^*$,
$$\{ (q,w_{U,v}w) | (q,w) \in P_{U,v} \} = \bigcup_{(q,w) \in U} \{ (q',ww_{\{ (q,\epsilon)\},v}w') | (q',w') \in P_{\{ (q,\epsilon)\},v} \},$$
one can prove that $\mathcal{W}(\tra)$ is equivalent to $\mathcal{D}(\tra)$.
The result then follows from the previous corollary.
\end{proof}

\begin{corollary}\label{pref_empt}
Let $U \in \mathcal{P}_f(Q, \Gamma^*)$ be the target of a non empty run in $\bar{\mathcal{W}}(\tra)$.
Then the largest common prefix of the words $\{ w | \exists q \in Q \textup{ s.t. } (q,w) \in R_{U,v} \}$ is equal to $\epsilon$.
\end{corollary}

\begin{proof}
By definition of $\bar{\mathcal{W}}(\tra)$, either $U = \{(q,\epsilon)\}$ for some $q \in Q$, and the results follows immediately, or $U$ is the target of a run in $\bar{\mathcal{D}}(\tra)$, and the result follows from lemma \ref{run_equ}.
\end{proof}

\paragraph{Proof of Proposition \ref{weakdet}}

\begin{proof}
By corollary \ref{wdet_equ}, $\mathcal{W}(\tra)$ and $\tra$ are equivalent.
Suppose that $\tra$ is weakly twinned.
We shall show that the subset $Q'$ of $\mathcal{P}_f(Q \times \Gamma^*)$ accessible from $U_0$ in $\bar{\mathcal{W}}(\tra)$ is finite.\\
Let $Q_b$ be the subset of $Q'$ composed of the elements $U \in Q'$ that contain at least two pairs $(q_1,v_1)$, $(q_2,v_2)$ such that $|\Delta(v_1,v_2)| \geq 2 M_{\tra} |Q|^3$.
We shall prove that $Q_b$ is empty, which proves the desired result, as it bounds the size of the words present in the sets $U \in Q'$ : if there were an element $U$ of $Q'$ containing a word of size greater than $2 M_{\tra} |Q|^3$, then it would be in $Q'$, by corollary \ref{pref_empt}.\\
Suppose ab absurdo that $Q_b$ is not empty.
Let $r : U_0 \xrightarrow{u|v} V \in Q_b$ be a run in $\mathcal{W}(\tra)$ such that for every $U \in Q_b$, for every run $r_{U} : U_0 \xrightarrow{x|y} U$ in $\mathcal{W}(\tra)$, the length of $r$ is shorter than or equal to the length of $r_{U}$.
The run $r$ can be decomposed into two parts
$r : U_0 \xrightarrow{u_0|v_0} V_0 \xrightarrow{\bar{u}|\bar{v}} V$
where $V_0$ is the first state such that the part of the run following it preserves the rank.
Therefore, the second part of the run is a run of $\bar{\mathcal{D}}(\tra)$.
Note that either $V_0 = U_0$, or its rank is different than the one of its predecessor.
In both cases, by definition of $\bar{\mathcal{W}}(\tra)$, $V_0 = \{ (p,\epsilon) \}$ for some $p \in Q$.
As $n_{V_0} = n_{V}$, there exists a pair $(p_0,v_0)$ in $V$ such that there exists a run from $p_0$ to $q$ in $\tra$.
Moreover, by supposition, $V$ contains two pairs $(p_1,v_1)$, $(p_2,v_2)$ such that $\Delta(v_1,v_2) \geq M_{\tra} |Q|^3$.
Then, by lemma \ref{run_equ}, $\tra$ admits three runs 
$$\begin{array}{lll}
r_0 & : & p = p_0^0 \xrightarrow{u^1|v_0^1} p_0^1 \xrightarrow{u^2|v_0^2} \ldots \xrightarrow{u^n|v_0^{n}} p_0^{n} = p_0,\\
r_1 & : & p = p_1^0 \xrightarrow{u^1|v_1^1} p_1^1 \xrightarrow{u^2|v_1^2} \ldots \xrightarrow{u^n|v_1^{n}} p_1^{n} = p_1,\\
r_2 & : & p = p_2^0 \xrightarrow{u^1|v_2^1} p_2^1 \xrightarrow{u^2|v_2^2} \ldots \xrightarrow{u^n|v_2^{n}} p_2^{n} = p_2,
\end{array}$$
such that $u^{1} \ldots u^{n} = \bar{u}$, and for every $0 \leq i < n$, for every $0 \leq j \leq 2$, $p_j^i \xrightarrow{u^{i+1}|v_j^{i+1}} p_j^{i+1}$ is a transition of $\tra$ and $v_j^{1} \ldots v_j^{n} = \bar{v}v_j$.
Then $n \geq |Q|^3$, because $|v_1^{1} \ldots v_1^{n}| + |v_2^{1} \ldots v_2^{n}| \geq \Delta(v_1,v_2) \geq 2 M_{\tra} |Q|^3$ by supposition, and for every $0 \leq i \leq n$, $|v_1^{i}| + |v_2^i| \leq 2 M_{\tra}$ by definition of $M_{\tra}$.
Therefore there exist $0 \leq i_1<i_2 \leq n$ such that for every $0 \leq j \leq 2$, $p_j^{i_1} = p_j^{i_2}$.\\
If $\Delta(v_1^{1} \ldots v_1^{i_1},v_2^{1} \ldots v_2^{i_1}) = \Delta(v_1^{1} \ldots v_1^{i_2},v_2^{1} \ldots v_2^{i_2})$, then 
$$\Delta(v_1^{1} \ldots v_1^{i_1}v_1^{i_2} \ldots v_1^{n},v_2^{1} \ldots v_2^{i_1}v_2^{i_2} \ldots v_2^{n}) = \Delta(v_1^{1} \ldots v_1^{n},v_2^{1} \ldots v_2^{n}) \geq  2M_{\tra} |Q|^3,$$
hence the run in $\mathcal{W}(\tra)$ corresponding to the runs of $\tra$ obtained by removing the part between $i_1$ and $i_2$ of $r_1$ and $r_2$ is a run strictly shorter than $r$ between $U_0$ and an element of $Q_b$, which contradicts the minimality of $r$.\\
If $\Delta(v_1^{1} \ldots v_1^{i_1},v_2^{1} \ldots v_2^{i_1}) \neq \Delta(v_1^{1} \ldots v_1^{i_2},v_2^{1} \ldots v_2^{i_2})$, then $\Delta(v_0^{1} \ldots v_0^{i_1},v_j^{1} \ldots v_j^{i_1}) \neq \Delta(v_0^{1} \ldots v_0^{i_2},v_j^{1} \ldots v_j^{i_2})$, for $j = 1$ or $j = 2$.
Moreover, by choice of $p_0$, there exists a run $p_0 \xrightarrow{x_0|y_0} p$.
Therefore there is a contradiction with the fact that $\tra$ is weakly twinned, exposed by taking the runs
$$\begin{array}{llll}
p_0 \xrightarrow{x_0|y_0} p \xrightarrow{u^{1} \ldots u^{i_1}|v_j^{1} \ldots v_j^{i_1}} p_j^{i_1}, & \ & p_0^{i_1}  \xrightarrow{u^{i_1} \ldots u^{i_2}|v_0^{i_1} \ldots v_0^{i_2}} p_0^{i_1},\\
p_0 \xrightarrow{x_0|y_0} p \xrightarrow{u^{1} \ldots u^{i_1}|v_0^{1} \ldots v_0^{i_1}} p_0^{i_1}, & \ & p_j^{i_1} \xrightarrow{u^{i_1} \ldots u^{i_2}|v_j^{i_1} \ldots v_j^{i_2}} p_j^{i_1}.
\end{array}$$
This concludes the proof.
\end{proof}

\subsection*{Proofs of section 4}

\paragraph{Proof of lemma \ref{thm:streaming}}

\begin{proof}
    If $R$ is multi-sequential, then 
    $R = \bigcup_{i=1}^k \inter{\mathcal{D}_i}$ for $k$ sequential
    transducers $\mathcal{D}_1,\dots,\mathcal{D}_k$. Let us show that $R$ is constant
    memory computable, i.e., let us construct $f,M,K$ as in the
    definition of constant memory computability. The function $f$
    associates with $u$ a word $b_1\dots b_k\in \{0,1\}^*$ such that 
    $b_i = 1$ iff $u\in dom(\mathcal{D}_i)$. The machine $M$ 
    simulates all the $\mathcal{D}_i$ such that $b_i = 1$
    successively. It uses at most $k$ output tapes to write 
    the output words $\inter{\mathcal{D}_i}(u)$. It does not need to
    use the working tape (and therefore we can take $K=0$) because it
    can simulate the states of $\mathcal{D}_i$ by using its internal
    states. A state of $M$ is therefore a tuple of states from distinct $\mathcal{D}_i$.
     $M$ uses $f(u)$ to initialise the first tuple of
    states to $(q_0^i)_{i\in \{ j\ |
      b_j=1\}}$, where $q_0^i$ is the initial state of $\mathcal{D}_i$.
     $M$ stops when it reads the first blank symbol on the
    input tape. 

    Conversely, let us show that any constant memory computable
    relation is multi-sequential. Let $f,k,M,K$ as in the definition
    of  constant memory computability. Let $P = \{ f(u)\ |\ u\in
    \Sigma^*\}$. Since for all $u\in\Sigma^*$, $|f(u)|\leq k$, the set
    $P$ is finite. Given $v\in P$, the machine $M$ reading input of
    the form $v\# u$ can be seen as the union of $k$ sequential
    transducers $\mathcal{D}_{1,v},\dots \mathcal{D}_{k,v}$. Informally, since the memory used on the working tape
    is bounded by $K$, there are only a constant number of possible
    configurations on that tape. The states of $\mathcal{D}_{i,v}$ are
    these configurations. Then, the transducer $\mathcal{D}_{i,v}$
    simulates $M$ and only produces the symbols produced by $M$ on the
    $i^{\textup{th}}$ output tape. The accepting states of $\mathcal{D}_{i,v}$ are the
    configurations where $M$ halts. Note that since $M$ can perform
    several transitions without reading any input symbol, the resulting transducers $\mathcal{D}_{i,v}$ may
    have $\epsilon$-transitions. It is easily shown that the
    $\mathcal{D}_{i,v}$ can be turned into proper transducers without
    $\epsilon$-transitions, as shown in
    \cite{DBLP:conf/fsttcs/FiliotGRS11} for constant memory computable
    functions. Finally, we have $R = \bigcup_{v\in P}\bigcup_{i=1}^k
    \inter{\mathcal{D}_{i,v}}$, and therefore $R$ is multi-sequential.

    \eof
\end{proof}

\subsection*{Examples}
Here are two more examples of weakly twinned transducers that are not twinned.\\

The figure $\ref{fig:det2} $ presents the transducer $\tra_1$ along with its determinisation and weak determinisation.
It is not twinned, as exposed by the runs 
$$\begin{array}{lllll}
q_0 & \xrightarrow{aaaa|a} & q_5 & \xrightarrow{ba|a} & q_5,\\
q_0 & \xrightarrow{aaaa|b} & q_5 & \xrightarrow{ba|a} & q_5.
\end{array}$$
Therefore, the determinisation is infinite.
However, the weak determinisation is finite.
The dotted edge of $\mathcal{D}(\tra_1)$, definitely leaving the SCC $\{ q_0,q_1,q_2 \}$ of $\tra_1$, is split into the three dotted edges of $\mathcal{W}(\tra_1)$.
The figure $\ref{fig:det2b}$ exposes the decomposition of $\tra_1$ into a union of sequential transducers.\\ 

The figure $\ref{fig:det3} $ present the transducer $\tra_2$ along with its determinisation and weak determinisation.
It is not twinned, as exposed by the runs 
$$\begin{array}{lllll}
q_0 & \xrightarrow{aaaaa|baaaa} & q_6 & \xrightarrow{a|a} & q_6,\\
q_0 & \xrightarrow{aaaaa|abaaa} & q_6 & \xrightarrow{a|a} & q_6.
\end{array}$$
Therefore, the determinisation is infinite.
However, the weak determinisation is finite.
The dotted edge of $\mathcal{D}(\tra_2)$, definitely leaving the SCC $\{ q_0,q_1,q_2,q_3 \}$ of $\tra_1$, is split into the three dotted edges of $\mathcal{W}(\tra_2)$.
The figure $\ref{fig:det3b}$ exposes the decomposition of $\tra_2$ into a union of sequential transducers.\\ 

\begin{figure}[!ht]
\subfigure[$\tra_1$\label{ex21}]{
\begin{tikzpicture}[->,>=stealth',auto,node distance=2cm,thick,scale=0.7,every node/.style={scale=0.7}]
  \tikzstyle{every state}=[fill=yellow!30,text=black]
  \tikzstyle{initial}=[initial by arrow, initial where=left, initial text=]
  \tikzstyle{accepting}=[accepting by arrow, accepting where=right, accepting text=]

  \node[initial,state] (A)  {$q_0$};
  \node[state] [above right of=A] (B)  {$q_1$};
  \node[state] [below right of=B] (C)  {$q_2$};
  \node[state] [right of=C] (D)  {$q_3$};
  \node[state] [above right of=D] (E)  {$q_4$};
  \node[accepting,state] [below right of=E] (F)  {$q_5$};
  \path (A) edge node {\trans{a}{a}} (B);
  \path (B) edge node {\trans{a}{\epsilon}} (C);
  \path (A) edge [bend right] node {\trans{a}{b}} (C);
  \path (C) edge [bend right] node {\trans{b}{\epsilon}} (A);
  \path (C) edge node {\trans{a}{\epsilon}} (D);
  \path (D) edge node {\trans{a}{\epsilon}} (E);
  \path (D) edge [bend right] node {\trans{a}{\epsilon}} (F);
  \path (E) edge node {\trans{a}{\epsilon}} (F);
  \path (F) edge [bend right] node {\trans{b}{a}} (D);
\end{tikzpicture}
}

\subfigure[$\mathcal{D}(\tra_1)$\label{ex22}]{
\begin{tikzpicture}[->,>=stealth',auto,node distance=2cm,thick,scale=0.7,every node/.style={scale=0.7}]
  \tikzstyle{every state}=[fill=yellow!30,text=black, shape = rectangle,rounded corners]
  \tikzstyle{initial}=[initial by arrow, initial where=left, initial text=]
  \tikzstyle{accepting}=[accepting by arrow, accepting where=right, accepting text=]

  \node[initial,state] (A)  {$q_0 : \epsilon$};
  \node[state] [above right of=A] (B)  {$\begin{array}{lll}q_1 & : & a\\ q_2 & : & b\end{array}$};
  \node[state] [below right of=B] (C)  {$\begin{array}{lll}q_2 & : & a\\ q_3 & : & b\end{array}$};
  \node[state] [right of=C] (D)  {$\begin{array}{lll}q_3 & : & a\\  q_4 & : & b \\  q_5 & : & b \end{array}$};
  \node[state] [right of=D] (D1)  {$\begin{array}{lll}q_3 & : & \epsilon \end{array}$};
  \node[state] [above right of=D1] (D2)  {$\begin{array}{lll}q_4 & : & \epsilon \end{array}$};
  \node[accepting,state] [below right of=D2] (D3)  {$\begin{array}{lll}q_5 & : & \epsilon \end{array}$};
  \node[state] [below of=D] (E)  {$\begin{array}{lll}q_4 & : & a,\\ q_5 & : & a,\\  &  & b \end{array}$};

  \tikzstyle{accepting}=[accepting by arrow, accepting where=below, accepting text=\vdots]

  \node[accepting,state] [below of=E] (F)  {$\begin{array}{lll}q_3 & : & aa,\\  &  & ba\end{array}$};
  \path (B) edge node {\trans{a}{\epsilon}} (C);
  \path (A) edge [bend left] node {\trans{a}{\epsilon}} (B);
  \path (B) edge [bend left] node {\trans{b}{b}} (A);
  \path (C) edge node {\trans{b}{a}} (A);
  \path (C) edge [densely dotted] node {\trans{a}{\epsilon}} (D);
  \path (D) edge node {\trans{b}{ba}} (D1);
  \path (D1) edge node {\trans{a}{\epsilon}} (D2);
  \path (D1) edge [bend right] node {\trans{a}{\epsilon}} (D3);
  \path (D2) edge node {\trans{a}{\epsilon}} (D3);
  \path (D3) edge [bend right] node {\trans{b}{a}} (D1);
  \path (D) edge node {\trans{a}{\epsilon}} (E);
  \path (E) edge [bend right] node {\trans{a}{\epsilon}} (D3);
  \path (E) edge node {\trans{b}{\epsilon}} (F);
\end{tikzpicture}
}

\subfigure[$\mathcal{W}(\tra_1)$\label{ex23}]{
\begin{tikzpicture}[->,>=stealth',auto,node distance=2cm,thick,scale=0.7,every node/.style={scale=0.7}]
  \tikzstyle{every state}=[fill=yellow!30,text=black, shape = rectangle,rounded corners]
  \tikzstyle{initial}=[initial by arrow, initial where=left, initial text=]
  \tikzstyle{accepting}=[accepting by arrow, accepting where=right, accepting text=]

  \node[initial,state] (A)  {$q_0 : \epsilon$};
  \node[state] [above right of=A] (B)  {$\begin{array}{lll}q_1 & : & a\\ q_2 & : & b\end{array}$};
  \node[state] [below right of=B] (C)  {$\begin{array}{lll}q_2 & : & a\\ q_3 & : & b\end{array}$};
  \node[state] [right of=C] (D)  {$\begin{array}{lll}q_3 & : & \epsilon \end{array}$};
  \node[state] [above right of=D] (E)  {$\begin{array}{lll}q_4 & : & \epsilon \end{array}$};
  \node[state,accepting] [below right of=E] (F)  {$\begin{array}{lll}q_5 & : & \epsilon \end{array}$};
  \path (B) edge node {\trans{a}{\epsilon}} (C);
  \path (A) edge [bend left] node {\trans{a}{\epsilon}} (B);
  \path (B) edge [bend left] node {\trans{b}{b}} (A);
  \path (C) edge node {\trans{b}{a}} (A);
  \path (C) edge [densely dotted] node {\trans{a}{a}} (D);
  \path (C) edge [densely dotted, bend left] node {\trans{a}{b}} (E);
  \path (C) edge [densely dotted, bend right] node {\trans{a}{b}} (F);
  \path (D) edge node {\trans{a}{\epsilon}} (E);
  \path (D) edge [bend right] node {\trans{a}{\epsilon}} (F);
  \path (E) edge node {\trans{a}{\epsilon}} (F);
  \path (F) edge [bend right] node {\trans{b}{a}} (D);
\end{tikzpicture}
}
\caption{Multisequential transducer that is not sequential.}
\label{fig:det2} 
\end{figure}
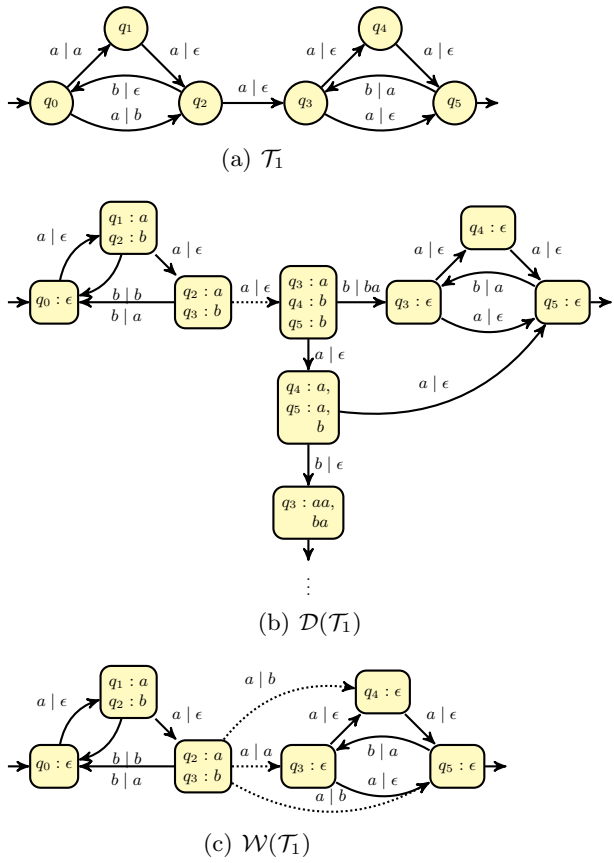

\begin{figure}[!ht]
\subfigure{
\begin{tikzpicture}[->,>=stealth',auto,node distance=2cm,thick,scale=0.7,every node/.style={scale=0.7}]
  \tikzstyle{every state}=[fill=yellow!30,text=black, shape = rectangle,rounded corners]
  \tikzstyle{initial}=[initial by arrow, initial where=left, initial text=]
  \tikzstyle{accepting}=[accepting by arrow, accepting where=right, accepting text=]

  \node[initial,state] (A)  {$q_0 : \epsilon$};
  \node[state] [above right of=A] (B)  {$\begin{array}{lll}q_1 & : & a\\ q_2 & : & b\end{array}$};
  \node[state] [below right of=B] (C)  {$\begin{array}{lll}q_2 & : & a\\ q_3 & : & b\end{array}$};
  \node[state] [right of=C] (D)  {$\begin{array}{lll}q_3 & : & \epsilon \end{array}$};
  \node[state] [above right of=D] (E)  {$\begin{array}{lll}q_4 & : & \epsilon \end{array}$};
  \node[state,accepting] [below right of=E] (F)  {$\begin{array}{lll}q_5 & : & \epsilon \end{array}$};
  \path (B) edge node {\trans{a}{\epsilon}} (C);
  \path (A) edge [bend left] node {\trans{a}{\epsilon}} (B);
  \path (B) edge [bend left] node {\trans{b}{b}} (A);
  \path (C) edge node {\trans{b}{a}} (A);
  \path (C) edge [densely dotted] node {\trans{a}{a}} (D);
  \path (D) edge node {\trans{a}{\epsilon}} (E);
  \path (D) edge [bend right] node {\trans{a}{\epsilon}} (F);
  \path (E) edge node {\trans{a}{\epsilon}} (F);
  \path (F) edge [bend right] node {\trans{b}{a}} (D);
\end{tikzpicture}
}

\subfigure{
\begin{tikzpicture}[->,>=stealth',auto,node distance=2cm,thick,scale=0.7,every node/.style={scale=0.7}]
  \tikzstyle{every state}=[fill=yellow!30,text=black, shape = rectangle,rounded corners]
  \tikzstyle{initial}=[initial by arrow, initial where=left, initial text=]
  \tikzstyle{accepting}=[accepting by arrow, accepting where=right, accepting text=]

  \node[initial,state] (A)  {$q_0 : \epsilon$};
  \node[state] [above right of=A] (B)  {$\begin{array}{lll}q_1 & : & a\\ q_2 & : & b\end{array}$};
  \node[state] [below right of=B] (C)  {$\begin{array}{lll}q_2 & : & a\\ q_3 & : & b\end{array}$};
  \node[state] [right of=C] (D)  {$\begin{array}{lll}q_3 & : & \epsilon \end{array}$};
  \node[state] [above right of=D] (E)  {$\begin{array}{lll}q_4 & : & \epsilon \end{array}$};
  \node[state,accepting] [below right of=E] (F)  {$\begin{array}{lll}q_5 & : & \epsilon \end{array}$};
  \path (B) edge node {\trans{a}{\epsilon}} (C);
  \path (A) edge [bend left] node {\trans{a}{\epsilon}} (B);
  \path (B) edge [bend left] node {\trans{b}{b}} (A);
  \path (C) edge node {\trans{b}{a}} (A);
  \path (C) edge [densely dotted, bend left] node {\trans{a}{b}} (E);
  \path (D) edge node {\trans{a}{\epsilon}} (E);
  \path (D) edge [bend right] node {\trans{a}{\epsilon}} (F);
  \path (E) edge node {\trans{a}{\epsilon}} (F);
  \path (F) edge [bend right] node {\trans{b}{a}} (D);
\end{tikzpicture}
}

\subfigure{
\begin{tikzpicture}[->,>=stealth',auto,node distance=2cm,thick,scale=0.7,every node/.style={scale=0.7}]
  \tikzstyle{every state}=[fill=yellow!30,text=black, shape = rectangle,rounded corners]
  \tikzstyle{initial}=[initial by arrow, initial where=left, initial text=]
  \tikzstyle{accepting}=[accepting by arrow, accepting where=right, accepting text=]

  \node[initial,state] (A)  {$q_0 : \epsilon$};
  \node[state] [above right of=A] (B)  {$\begin{array}{lll}q_1 & : & a\\ q_2 & : & b\end{array}$};
  \node[state] [below right of=B] (C)  {$\begin{array}{lll}q_2 & : & a\\ q_3 & : & b\end{array}$};
  \node[state] [right of=C] (D)  {$\begin{array}{lll}q_3 & : & \epsilon \end{array}$};
  \node[state] [above right of=D] (E)  {$\begin{array}{lll}q_4 & : & \epsilon \end{array}$};
  \node[state,accepting] [below right of=E] (F)  {$\begin{array}{lll}q_5 & : & \epsilon \end{array}$};
  \path (B) edge node {\trans{a}{\epsilon}} (C);
  \path (A) edge [bend left] node {\trans{a}{\epsilon}} (B);
  \path (B) edge [bend left] node {\trans{b}{b}} (A);
  \path (C) edge node {\trans{b}{a}} (A);
  \path (C) edge [densely dotted, bend right] node {\trans{a}{b}} (F);
  \path (D) edge node {\trans{a}{\epsilon}} (E);
  \path (D) edge [bend right] node {\trans{a}{\epsilon}} (F);
  \path (E) edge node {\trans{a}{\epsilon}} (F);
  \path (F) edge [bend right] node {\trans{b}{a}} (D);
\end{tikzpicture}
}
\caption{The transducer $\textsf{trim}(\textsf{split}(\mathcal{W}(\tra_1)))$}
\label{fig:det2b} 
\end{figure}

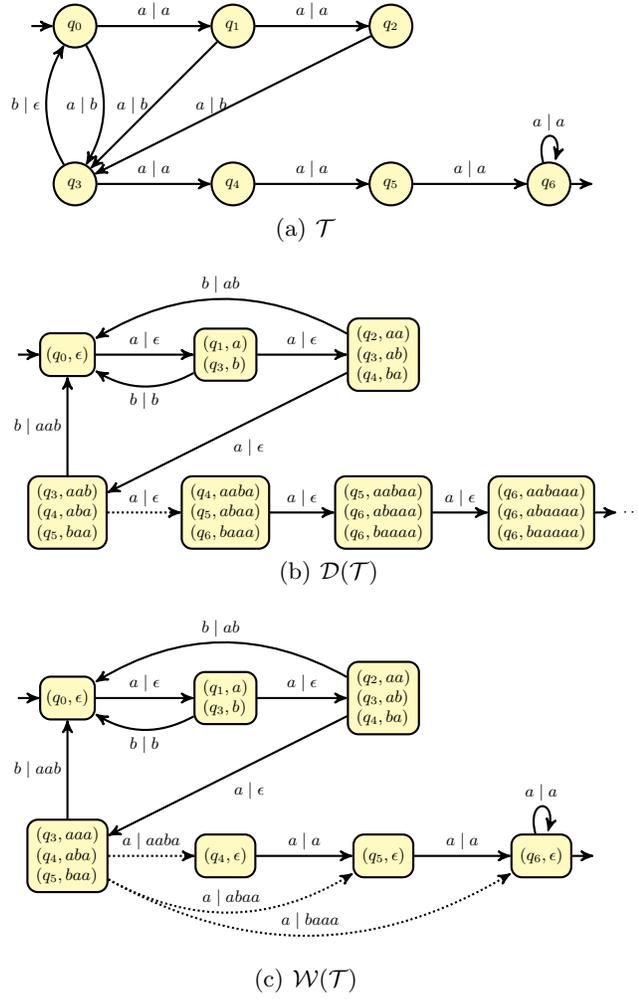
\begin{figure}[!ht]
\subfigure[$\tra$\label{ex31}]{
\begin{tikzpicture}[->,>=stealth',auto,node distance=3cm,thick,scale=0.7,every node/.style={scale=0.7}]
  \tikzstyle{every state}=[fill=yellow!30,text=black]
  \tikzstyle{initial}=[initial by arrow, initial where=left, initial text=]
  \tikzstyle{accepting}=[accepting by arrow, accepting where=right, accepting text=]

  \node[initial,state] (A)  {$q_0$};
  \node[state] [right of=A] (B)  {$q_1$};
  \node[state] [right of=B] (C)  {$q_2$};
  \node[state] [below of=A] (D)  {$q_3$};
  \node[state] [right of=D] (E)  {$q_4$};
  \node[state] [right of=E] (F)  {$q_5$};
  \node[accepting,state] [right of=F] (G)  {$q_6$};
  \path (A) edge node {\trans{a}{a}} (B);
  \path (B) edge node {\trans{a}{a}} (C);
  \path (A) edge [left, bend left] node {\trans{a}{b}} (D);
  \path (B) edge [left] node {\trans{a}{b}} (D);
  \path (C) edge [left] node {\trans{a}{b}} (D);
  \path (D) edge [left, bend left] node {\trans{b}{\epsilon}} (A);
  \path (D) edge node {\trans{a}{a}} (E);
  \path (E) edge node {\trans{a}{a}} (F);
  \path (F) edge node {\trans{a}{a}} (G);
  \path (G) edge [loop above] node {\trans{a}{a}} (G);
\end{tikzpicture}
}

\subfigure[$\mathcal{D}(\tra)$\label{ex32}]{
\begin{tikzpicture}[->,>=stealth',auto,node distance=3cm,thick,scale=0.7,every node/.style={scale=0.7}]
  \tikzstyle{every state}=[fill=yellow!30,text=black, shape = rectangle,rounded corners]
  \tikzstyle{initial}=[initial by arrow, initial where=left, initial text=]
  \tikzstyle{accepting}=[accepting by arrow, accepting where=right, accepting text=$\cdots$]

  \node[initial,state] (A)  {$(q_0, \epsilon)$};
  \node[state] [right of=A] (B)  {$\begin{array}{l}(q_1, a)\\ (q_3, b)\end{array}$};
  \node[state] [right of=B] (C)  {$\begin{array}{lll}(q_2, aa)\\ (q_3, ab)\\ (q_4, ba)\end{array}$};
  \node[state] [below of=A] (D)  {$\begin{array}{lll}(q_3, aab)\\ (q_4, aba)\\ (q_5, baa)\end{array}$};
  \node[state] [right of=D] (E)  {$\begin{array}{lll}(q_4, aaba)\\ (q_5, abaa)\\ (q_6, baaa)\end{array}$};
  \node[state] [right of=E] (F)  {$\begin{array}{lll}(q_5, aabaa)\\ (q_6, abaaa)\\ (q_6, baaaa)\end{array}$};
  \node[accepting,state] [right of=F] (G)  {$\begin{array}{lll}(q_6, aabaaa)\\ (q_6, abaaaa)\\ (q_6, baaaaa)\end{array}$};
  \path (A) edge node {\trans{a}{\epsilon}} (B);
  \path (B) edge node {\trans{a}{\epsilon}} (C);
  \path (C) edge node {\trans{a}{\epsilon}} (D);
  \path (B) edge [bend left] node {\trans{b}{b}} (A);
  \path (C) edge [above,bend right] node {\trans{b}{ab}} (A);
  \path (D) edge node {\trans{b}{aab}} (A);
  \path (D) edge [densely dotted] node {\trans{a}{\epsilon}} (E);
  \path (E) edge node {\trans{a}{\epsilon}} (F);
  \path (F) edge node {\trans{a}{\epsilon}} (G);
\end{tikzpicture}
}

\subfigure[$\mathcal{W}(\tra)$\label{ex33}]{
\begin{tikzpicture}[->,>=stealth',auto,node distance=3cm,thick,scale=0.7,every node/.style={scale=0.7}]
  \tikzstyle{every state}=[fill=yellow!30,text=black, shape = rectangle,rounded corners]
  \tikzstyle{initial}=[initial by arrow, initial where=left, initial text=]
  \tikzstyle{accepting}=[accepting by arrow, accepting where=right, accepting text=]

  \node[initial,state] (A)  {$(q_0, \epsilon)$};
  \node[state] [right of=A] (B)  {$\begin{array}{l}(q_1, a)\\ (q_3, b)\end{array}$};
  \node[state] [right of=B] (C)  {$\begin{array}{lll}(q_2, aa)\\ (q_3, ab)\\ (q_4, ba)\end{array}$};
  \node[state] [below of=A] (D)  {$\begin{array}{lll}(q_3, aaa)\\ (q_4, aba)\\ (q_5, baa)\end{array}$};
  \node[state] [right of=D] (E)  {$\begin{array}{lll}(q_4, \epsilon)\end{array}$};
  \node[state] [right of=E] (F)  {$\begin{array}{lll}(q_5, \epsilon)\end{array}$};
  \node[accepting,state] [right of=F] (G)  {$\begin{array}{lll}(q_6, \epsilon)\end{array}$};
  \path (A) edge node {\trans{a}{\epsilon}} (B);
  \path (B) edge node {\trans{a}{\epsilon}} (C);
  \path (C) edge node {\trans{a}{\epsilon}} (D);
  \path (B) edge [bend left] node {\trans{b}{b}} (A);
  \path (C) edge [above,bend right] node {\trans{b}{ab}} (A);
  \path (D) edge node {\trans{b}{aab}} (A);
  \path (D) edge [densely dotted] node {\trans{a}{aaba}} (E);
  \path (D) edge [densely dotted, bend right] node {\trans{a}{abaa}} (F);
  \path (D) edge [densely dotted, bend right] node {\trans{a}{baaa}} (G);
  \path (E) edge node {\trans{a}{a}} (F);
  \path (F) edge node {\trans{a}{a}} (G);
  \path (G) edge [loop above] node {\trans{a}{a}} (G);
\end{tikzpicture}
}
\caption{Multisequential transducer that is not sequential.}
\label{fig:det3} 
\end{figure}

\begin{figure}[!ht]
\subfigure{
\begin{tikzpicture}[->,>=stealth',auto,node distance=3cm,thick,scale=0.7,every node/.style={scale=0.7}]
  \tikzstyle{every state}=[fill=yellow!30,text=black, shape = rectangle,rounded corners]
  \tikzstyle{initial}=[initial by arrow, initial where=left, initial text=]
  \tikzstyle{accepting}=[accepting by arrow, accepting where=right, accepting text=]

  \node[initial,state] (A)  {$(q_0, \epsilon)$};
  \node[state] [right of=A] (B)  {$\begin{array}{l}(q_1, a)\\ (q_3, b)\end{array}$};
  \node[state] [right of=B] (C)  {$\begin{array}{lll}(q_2, aa)\\ (q_3, ab)\\ (q_4, ba)\end{array}$};
  \node[state] [below of=A] (D)  {$\begin{array}{lll}(q_3, aaa)\\ (q_4, aba)\\ (q_5, baa)\end{array}$};
  \node[state] [right of=D] (E)  {$\begin{array}{lll}(q_4, \epsilon)\end{array}$};
  \node[state] [right of=E] (F)  {$\begin{array}{lll}(q_5, \epsilon)\end{array}$};
  \node[accepting,state] [right of=F] (G)  {$\begin{array}{lll}(q_6, \epsilon)\end{array}$};
  \path (A) edge node {\trans{a}{\epsilon}} (B);
  \path (B) edge node {\trans{a}{\epsilon}} (C);
  \path (C) edge node {\trans{a}{\epsilon}} (D);
  \path (B) edge [bend left] node {\trans{b}{b}} (A);
  \path (C) edge [above, bend right] node {\trans{b}{ab}} (A);
  \path (D) edge node {\trans{b}{aab}} (A);
  \path (D) edge [densely dotted] node {\trans{a}{aaba}} (E);
  \path (E) edge node {\trans{a}{a}} (F);
  \path (F) edge node {\trans{a}{a}} (G);
  \path (G) edge [loop above] node {\trans{a}{a}} (G);
\end{tikzpicture}
}

\subfigure{
\begin{tikzpicture}[->,>=stealth',auto,node distance=3cm,thick,scale=0.7,every node/.style={scale=0.7}]
  \tikzstyle{every state}=[fill=yellow!30,text=black, shape = rectangle,rounded corners]
  \tikzstyle{initial}=[initial by arrow, initial where=left, initial text=]
  \tikzstyle{accepting}=[accepting by arrow, accepting where=right, accepting text=]

  \node[initial,state] (A)  {$(q_0, \epsilon)$};
  \node[state] [right of=A] (B)  {$\begin{array}{l}(q_1, a)\\ (q_3, b)\end{array}$};
  \node[state] [right of=B] (C)  {$\begin{array}{lll}(q_2, aa)\\ (q_3, ab)\\ (q_4, ba)\end{array}$};
  \node[state] [below of=A] (D)  {$\begin{array}{lll}(q_3, aaa)\\ (q_4, aba)\\ (q_5, baa)\end{array}$};
  \node [right of=D] (E)  {$ $};
  \node[state] [right of=E] (F)  {$\begin{array}{lll}(q_5, \epsilon)\end{array}$};
  \node[accepting,state] [right of=F] (G)  {$\begin{array}{lll}(q_6, \epsilon)\end{array}$};
  \path (A) edge node {\trans{a}{\epsilon}} (B);
  \path (B) edge node {\trans{a}{\epsilon}} (C);
  \path (C) edge node {\trans{a}{\epsilon}} (D);
  \path (B) edge [bend left] node {\trans{b}{b}} (A);
  \path (C) edge [above, bend right] node {\trans{b}{ab}} (A);
  \path (D) edge node {\trans{b}{aab}} (A);
  \path (D) edge [densely dotted] node {\trans{a}{abaa}} (F);
  \path (F) edge node {\trans{a}{a}} (G);
  \path (G) edge [loop above] node {\trans{a}{a}} (G);
\end{tikzpicture}
}

\subfigure{
\begin{tikzpicture}[->,>=stealth',auto,node distance=3cm,thick,scale=0.7,every node/.style={scale=0.7}]
  \tikzstyle{every state}=[fill=yellow!30,text=black, shape = rectangle,rounded corners]
  \tikzstyle{initial}=[initial by arrow, initial where=left, initial text=]
  \tikzstyle{accepting}=[accepting by arrow, accepting where=right, accepting text=]

  \node[initial,state] (A)  {$(q_0, \epsilon)$};
  \node[state] [right of=A] (B)  {$\begin{array}{l}(q_1, a)\\ (q_3, b)\end{array}$};
  \node[state] [right of=B] (C)  {$\begin{array}{lll}(q_2, aa)\\ (q_3, ab)\\ (q_4, ba)\end{array}$};
  \node[state] [below of=A] (D)  {$\begin{array}{lll}(q_3, aaa)\\ (q_4, aba)\\ (q_5, baa)\end{array}$};
  \node [right of=D] (E)  {$ $};
  \node [right of=E] (F)  {$ $};
  \node[accepting,state] [right of=F] (G)  {$\begin{array}{lll}(q_6, \epsilon)\end{array}$};
  \path (A) edge node {\trans{a}{\epsilon}} (B);
  \path (B) edge node {\trans{a}{\epsilon}} (C);
  \path (C) edge node {\trans{a}{\epsilon}} (D);
  \path (B) edge [bend left] node {\trans{b}{b}} (A);
  \path (C) edge [above, bend right] node {\trans{b}{ab}} (A);
  \path (D) edge node {\trans{b}{aab}} (A);
  \path (D) edge [densely dotted] node {\trans{a}{baaa}} (G);
  \path (G) edge [loop above] node {\trans{a}{a}} (G);
\end{tikzpicture}
}
\caption{The transducer $\textsf{trim}(\textsf{split}(\mathcal{W}(\tra_2)))$}
\label{fig:det3b} 
\end{figure}


\end{document}